\documentclass[preprint,12pt,authoryear]{elsarticle}

\usepackage{amssymb}
\usepackage{amsmath}
\usepackage{amsthm}
\usepackage{graphicx}
\usepackage{natbib}
\usepackage{hyperref}
\usepackage{geometry}
\usepackage{memhfixc}
\usepackage{tikz}
\usepackage{circuitikz}
\usetikzlibrary{arrows}
\usetikzlibrary{positioning}
\usetikzlibrary{automata, arrows.meta, positioning}
\usepackage{setspace}
\usepackage{xcolor} %colors
\usepackage{caption} %subfigure
\usepackage{subcaption}  %subfigure
\usepackage{booktabs} %tables
\usepackage{multirow} %tables
\usepackage{tablefootnote} %table footnote
\usepackage{float}
\usepackage{algorithm}
\usepackage{algpseudocode}
\definecolor{universityred}{RGB}{173,0,0} %red
\definecolor{secondarygray}{RGB}{139,157,161} 
\definecolor{universityblue}{RGB}{0,78,116}
\definecolor{lightgray}{RGB}{180, 180, 180}


\newtheorem{theorem}{Theorem}
\newtheorem{lemma}{Lemma}
\newtheorem{proposition}{Proposition}

\def\bs{\boldsymbol}
\def\th{\boldsymbol{\theta}}

\journal{Computational Statistics \& Data Analysis}

\begin{document}

\begin{frontmatter}



\title{Empirical-Bayes Elastic-Net Computation for Exponential Random Graph Models}

\author[label1]{Dan Han}
\author[label2]{Vicki Modisette}
\author[label3]{Ting Li}
\author[label1]{Akidul Haque}

\affiliation[label1]{organization={Department of Mathematics, University of Louisville},%Department and Organization
            addressline={2301 S 3rd St.}, 
            city={Louisville},
            postcode={40292}, 
            state={Kentucky},
            country={United States}}

\affiliation[label2]{organization={Department of Mathematics, Eastern Kentucky University},%Department and Organization
            city={Richmond},
            state={Kentucky},
            country={United States}}

\affiliation[label3]{organization={Department of Computer Science, Emory University},%Department and Organization
            city={Atlanta},
            state={Georgia},
            country={United States}}

\begin{abstract}
Exponential random graph models (ERGMs) describe dependence among network ties, but inference becomes difficult when the likelihood is intractable and candidate network statistics are strongly correlated. We introduce BERGM Elastic Net, an adaptive empirical-Bayes approach that combines lasso shrinkage with ridge stabilization in a Bayesian ERGM. A latent-variable formulation supports approximate exchange sampling, while empirical-Bayes updates adapt the amount of regularization to the observed network. We connect the proposed prior to elastic-net penalized likelihood and clarify the interpretation of thresholded reporting and coefficient grouping. The method is developed for over-specified network models containing many related structural and covariate effects.

\end{abstract}




\end{frontmatter}



\section{Introduction}\label{sec:intro}
Network data arise whenever scientific units are meaningfully connected: students nominate friends, proteins interact, individuals transmit infection, firms trade, and institutions exchange resources. In such settings, the primary object of inference is not merely a collection of independent observations, but a system of relational dependencies. The presence of one tie can alter the probability of other ties through mechanisms such as reciprocity, transitivity, degree heterogeneity, and homophily. This dependence is precisely what makes network data scientifically rich, but it also makes standard regression or independent-dyad models inadequate for many inferential questions \citep{newman2003structure,zhang2019social}.

Exponential random graph models (ERGMs) provide a principled likelihood-based framework for modeling these dependencies. Originating from log-linear and Markov random graph formulations \citep{FrankStrauss1986,WassermanPattison1996}, ERGMs assign probabilities to entire networks through a vector of sufficient statistics, such as the number of edges, shared partners, triangles, \(k\)-stars, and nodal-attribute mixing effects. This representation is attractive because each coefficient has a direct interpretation as the conditional contribution of a local network feature to tie formation. Modern ERGM methodology has substantially expanded the class of stable and interpretable specifications, including curved and geometrically weighted terms designed to capture transitivity and degree heterogeneity without relying solely on unstable raw triangle and star counts \citep{Snijders2006,Robins2007intro,Hunter2008ergm,lusher2013}.

Despite this progress, ERGM inference remains statistically and computationally challenging. The likelihood contains a normalizing constant that sums over all possible networks on the node set, making exact likelihood evaluation infeasible except in very small or highly restricted cases. Frequentist methods therefore rely on pseudolikelihood approximations or Monte Carlo maximum likelihood procedures \citep{strauss1990pseudolikelihood,Snijders2002}. Bayesian ERGM methods face the same intractable normalizing constant and have motivated exchange algorithms, approximate exchange samplers, and software implementations such as \texttt{Bergm} \citep{CaimoFriel2011,Jin2013,CaimoFriel2014}. In addition to computational intractability, ERGMs can exhibit near-degeneracy, sensitivity, and practical nonidentifiability, especially when dependence terms compete to explain the same structural pattern \citep{Handcock2003,Snijders2006,ChatterjeeDiaconis2013}.

A second, increasingly important difficulty is model specification. Applied ERGM analyses often begin with a large set of plausible endogenous statistics and exogenous covariate effects. For example, a model may include density, homophily, transitivity, degree terms, attribute effects, and many nodal or dyadic covariates. These candidate statistics are rarely orthogonal. Edges, degree effects, shared-partner terms, and triangles may be strongly correlated; nodal covariates may also be collinear when they measure related social or biological attributes. In such settings, unregularized ERGM estimates can have high variance, weak interpretability, and poor stability across candidate specifications. Model selection is therefore not a peripheral issue, but a central part of reliable ERGM inference.

Penalized estimation offers a natural route to this problem. In ordinary regression, the lasso encourages sparsity by shrinking some coefficients to zero \cite{tibshirani1996regression}, ridge regression stabilizes estimates under collinearity \citep{hoerl1970ridge}, and the elastic net combines both effects \citep{zou2005regularization}. Recent work has begun to translate these ideas into ERGM settings. Lasso-type ERGM procedures have been proposed for variable selection and ranking of candidate network statistics \citep{ButtazzoKauermann2026}, and hierarchical Bayesian adaptive-lasso and ridge ERGMs have been developed to perform variable selection while making posterior inference \citep{modisette2023penalized, Han2024adaptiveLasso}. These contributions demonstrate the value of penalization for high-dimensional network models, but they also reveal a limitation of purely \(L_1\)-based shrinkage: when active statistics or covariates are strongly correlated, lasso-type methods may select one representative from the group and suppress the others. Ridge-type shrinkage reduces this instability but does not by itself yield sparse model selection.

To address this gap, we develop a hierarchical elastic-net BERGM with an
explicit adaptive empirical-Bayes implementation. The proposed zero-centered
prior combines an \(L_1\) component, which pulls weak ERGM coefficients toward
zero, with an \(L_2\) component, which stabilizes estimation along correlated
coefficient directions. The latent representation connects the elastic-net
penalized full-likelihood estimator to a posterior mode, supplies conditionally
Gaussian prior components, introduces generalized inverse Gaussian local-scale
updates, and leads to Monte Carlo EM updates for the two penalty parameters.

The paper is positioned as a computational-statistics contribution. First, we
derive the proper elastic-net prior, its latent representation, posterior
propriety, and the Monte Carlo M-step. Second, we state an invariance theorem for
the ideal fixed-penalty population exchange--Gibbs kernel and a conditional
perturbation bound describing how finite auxiliary ERGM simulation affects that
ideal kernel. Third, we provide a decision-theoretic interpretation of the
thresholded reporting rule and a full-likelihood MAP coefficient-closeness
bound for related statistics. Fourth, we evaluate the complete implementation
in a reproducible correlated-signal simulation and illustrate it on the
\texttt{faux.magnolia.high} friendship network and an OpenAlex AI citation
subgraph.

The remainder of the paper is organized as follows.
Section~\ref{sec:model} introduces the ERGM likelihood, the elastic-net prior,
its latent representation, and the empirical-Bayes-motivated penalty updates.
Section~\ref{sec:algorithms} presents the approximate exchange--Gibbs
implementation and the fixed-penalty kernel analysis.
Section~\ref{sec:theory} develops the decision-theoretic reporting rule and the
pseudolikelihood and full-likelihood grouping results.
Section~\ref{sec:simulation} reports the
simulation study, and Section~\ref{sec:applications} presents the two network
applications. Section~\ref{sec:discussion} concludes with limitations and
directions for further work.

\section{Model Formulation}\label{sec:model}
\subsection{ERGMs and Elastic-Net Penalization}
Let \(Y\) denote a random network on \(n\) nodes, represented by an adjacency matrix \(Y=(Y_{ij})\), where \(Y_{ij}=1\) indicates the presence of a tie and \(Y_{ij}=0\) otherwise. For an undirected loopless network, the dyad set is \(\mathcal D=\{(i,j):i<j\}\); for a directed loopless network it is \(\mathcal D=\{(i,j):i\ne j\}\). The observed network is denoted by \(y\). For a vector of sufficient network statistics \(s(y)\in\mathbb R^p\), the ERGM likelihood has exponential-family form \citep{lusher2013}
\begin{equation}
 \pi (y \mid \th) = \frac{\exp\{\th^\top s(y)\}}{z(\th)},
 \qquad
 z(\th)=\sum_{y'\in\mathcal Y}\exp\{\th^\top s(y')\},
\end{equation}
where \(\mathcal Y\) is the set of all networks on the same node set. The log-likelihood is
\begin{equation}
  l(\th\mid y)=\th^\top s(y)-\log z(\th).
\end{equation}
Positive components of \(\th\) increase the probability of networks exhibiting the corresponding statistics, while negative components suppress those features.

When the sufficient statistics are highly correlated, the unpenalized likelihood may yield high-variance estimates. We therefore consider the elastic-net penalized ERGM estimator
\begin{equation}\label{prior}
\hat{\th}_{EN}
=
\arg\max_{\th}
\left\{l(\th\mid y)-P(\th)\right\},
\end{equation}
where
\[
P(\th)=\lambda_1\|\th\|_1+\lambda_2\|\th\|_2^2
=
\sum_{j=1}^p \lambda_1|\theta_j|+\lambda_2\theta_j^2,
\qquad \lambda_1,\lambda_2>0.
\]
The \(L_1\) component encourages shrinkage and sparse posterior modes, but the continuous prior does not create posterior point masses at zero. The \(L_2\) component stabilizes estimates and encourages coefficient similarity along correlated directions. Substitution into \eqref{prior} gives
\begin{equation}\label{likelihood}
\hat{\th}_{EN}
=
\arg\max_{\th}
\left[
\th^\top s(y)-\log z(\th)
-\sum_{j=1}^p\{\lambda_1|\theta_j|+\lambda_2\theta_j^2\}
\right].
\end{equation}

\subsection{Pseudolikelihood Representation}
For dyad-dependent ERGMs, the normalizing constant \(z(\th)\) is typically unavailable except for very small networks. The pseudolikelihood of \citet{strauss1990pseudolikelihood} replaces the full likelihood by the product of conditional dyad probabilities,
\begin{equation}
    PL(\th)= \prod_{(i,j)\in\mathcal D} P(Y_{ij}=y_{ij}\mid Y_{-ij}=y_{ij}^{c}, \th).
\end{equation}
For binary dyads this becomes
\begin{equation}
   PL(\th)= \prod_{(i,j)\in\mathcal D}
   P(Y_{ij}=1\mid y_{ij}^{c}, \th)^{y_{ij}}
   P(Y_{ij}=0\mid y_{ij}^{c}, \th)^{1-y_{ij}} .
\end{equation}
Let \(y_{ij}^{+}\) and \(y_{ij}^{-}\) denote the networks obtained by setting \(y_{ij}=1\) and \(y_{ij}=0\), respectively, and define the change statistic
\[
\Delta s_{ij}(y)=s(y_{ij}^{+})-s(y_{ij}^{-}).
\]
The conditional log-odds for dyad \((i,j)\) are
\begin{equation}\label{conditional log-odds}
    \log\left\{
    \frac{P(Y_{ij}=1\mid Y_{-ij}=y_{ij}^{c},\th)}
    {P(Y_{ij}=0\mid Y_{-ij}=y_{ij}^{c},\th)}
    \right\}
    =
    \th^\top\Delta s_{ij}(y).
\end{equation}
Thus
\begin{align}\label{logit}
P(Y_{ij}=1\mid y_{ij}^{c},\th)
&=
\frac{\exp\{\th^\top\Delta s_{ij}(y)\}}
{1+\exp\{\th^\top\Delta s_{ij}(y)\}}.
\end{align}
Taking logarithms yields the logistic-regression form
\begin{align}\label{likelihood of pseudolikelihood}
    \log PL(\th)
    =
    \sum_{(i,j)\in\mathcal D}
    \left[
    y_{ij}\th^\top\Delta s_{ij}(y)
    -
    \log\{1+\exp(\th^\top\Delta s_{ij}(y))\}
    \right].
\end{align}
Consequently, elastic-net penalized pseudolikelihood estimation is equivalent to elastic-net penalized logistic regression with design vectors \(\Delta s_{ij}(y)\). This equivalence is useful for motivation, screening, and initialization. The posterior sampler developed below does not target \(PL(\th)\); it targets the full ERGM likelihood in \eqref{posterior}.

\subsection{Bayesian Formulation}
With prior density \(\pi(\th)\), the ERGM posterior is
\begin{equation}\label{posterior}
    \pi(\th \mid y)
    =
    \frac{\pi(y\mid\th)\pi(\th)}{\pi(y)}
    =
    \frac{\exp\{\th^\top s(y)\}}{z(\th)}
    \frac{\pi(\th)}{\pi(y)}.
\end{equation}
This posterior is doubly intractable: the likelihood normalizing constant \(z(\th)\) and the model evidence \(\pi(y)\) are both unavailable in closed form \citep{lusher2013}. The model evidence is harmless for ratios, but the likelihood normalizing constant changes with \(\th\), so a direct Metropolis--Hastings update cannot be used without further work.

The obstruction appears in the ordinary Metropolis--Hastings ratio. To sample from
$
\pi(\th\mid \bs y)
\propto
\displaystyle\frac{\exp\{\th^\top s(\bs y)\}}{z(\th)}\pi(\th),
$
suppose the current state is \(\th\), and let
\(\th'\sim q(\cdot\mid \th)\) be a proposed value. The usual acceptance ratio is
\begin{equation}\label{eq:ordinary_MH_ratio}
R(\th,\th')
=
\frac{
\exp\{\th'^\top s(\bs y)\}z(\th')^{-1}\pi(\th')q(\th\mid \th')
}{
\exp\{\th^\top s(\bs y)\}z(\th)^{-1}\pi(\th)q(\th'\mid \th)
}.
\end{equation}
Equivalently,
\begin{equation}\label{eq:ordinary_MH_ratio_simplified}
R(\th,\th')
=
\exp\{(\th'-\th)^\top s(\bs y)\}
\frac{z(\th)}{z(\th')}
\frac{\pi(\th')}{\pi(\th)}
\frac{q(\th\mid \th')}{q(\th'\mid \th)}.
\end{equation}
The proposed value \(\th'\) is accepted with probability
\[
\alpha(\th,\th')=\min\{1,R(\th,\th')\}.
\]
The ratio \(z(\th)/z(\th')\) in
\eqref{eq:ordinary_MH_ratio_simplified} is the computational obstacle.

To avoid this intractable ratio, one can use the exchange algorithm
\citep{moller2006efficient,murray2006mcmc}.
After proposing \(\th'\sim q(\cdot\mid \th)\), an auxiliary network
\(\bs y'\) is generated from the ERGM likelihood at the proposed parameter,
\[
\bs y' \sim \pi(\cdot\mid \th')
=
\frac{\exp\{\th'^\top s(\bs y')\}}{z(\th')}.
\]
The exchange acceptance ratio is then
\begin{align}
R_{\mathrm{ex}}(\th,\th')
&=
\frac{
\exp\{\th'^\top s(\bs y)\}\pi(\th')q(\th\mid \th')
}{
\exp\{\th^\top s(\bs y)\}\pi(\th)q(\th'\mid \th)
}
\frac{
\exp\{\th^\top s(\bs y')\}
}{
\exp\{\th'^\top s(\bs y')\}
} \notag \\
&=
\exp\{(\th'-\th)^\top s(\bs y)+(\th-\th')^\top s(\bs y')\}
\frac{\pi(\th')}{\pi(\th)}
\frac{q(\th\mid \th')}{q(\th'\mid \th)} \notag \\
&=
\exp\{(\th-\th')^\top [s(\bs y')-s(\bs y)]\}
\frac{\pi(\th')}{\pi(\th)}
\frac{q(\th\mid \th')}{q(\th'\mid \th)}.
\label{eq:exchange_ratio}
\end{align}
The proposed value is accepted with probability
\[
\alpha_{\mathrm{ex}}(\th,\th')=\min\{1,R_{\mathrm{ex}}(\th,\th')\}.
\]
If the proposal density is symmetric, then
\(q(\th\mid \th')/q(\th'\mid \th)=1\), and the log-acceptance probability reduces to
\[
\log \alpha_{\mathrm{ex}}(\th,\th')
=
\min\left\{
0,\,
(\th-\th')^\top [s(\bs y')-s(\bs y)]
+
\log \frac{\pi(\th')}{\pi(\th)}
\right\}.
\]
In practice, exact simulation from \(\pi(\cdot\mid \th')\) is rarely available for
large ERGMs, and the auxiliary network \(\bs y'\) is typically generated by running
an internal MCMC sampler targeting \(\pi(\cdot\mid \th')\). This yields the
approximate exchange algorithm commonly used in Bayesian ERGM computation.

\subsection{Elastic-Net Prior and Posterior Mode}
Bridge regression \citep{Frank1993,Fu1998} defines the penalty in equation~\eqref{prior} as
\[
P(\th)=\|\th\|_{q}^{q}=\sum_{j=1}^{p}|\theta_j|^{q},
\]
which corresponds to the Bayesian posterior mode under the prior
\begin{equation}
    \pi_{\lambda,q}(\th)=C(\lambda,q)\exp\!\left(-\lambda\|\th\|_q^q\right).
\end{equation}
When \(q=1\), the model reduces to the lasso and induces a Laplace (double-exponential) prior, whereas \(q=2\) yields ridge regression and induces a Gaussian prior. The elastic-net penalty corresponds to the combined prior
\begin{equation}\label{penalty}
    \pi(\th)=C(\lambda_1,\lambda_2)e^{-\lambda_1\|\th\|_1-\lambda_2\|\th\|_2^2},
\end{equation}
where \(C(\lambda_1,\lambda_2)\) is the normalizing constant. The posterior distribution of \(\th\) is
\begin{equation}
    \pi(\th\mid \bs y)\propto \pi(\bs y\mid \th)\pi(\th).
\end{equation}
Let \(\hat{\th}_{MAP}\) denote the maximum a posteriori estimator under \eqref{penalty}. Then
\begin{align}\label{MAP}
   \hat{\th}_{MAP}
   &= \arg\max_{\th}\,\pi(\th\mid \bs y) \\
   &= \arg\max_{\th}\left[
   \frac{e^{\th^\top s(\bs y)}}{z(\th)}
   C(\lambda_1,\lambda_2)
   e^{-\lambda_1\|\th\|_1-\lambda_2\|\th\|_2^2}
   \right] \\
   &= \arg\max_{\th}\left[
   \th^\top s(\bs y)-\ln z(\th)-\lambda_1\|\th\|_1-\lambda_2\|\th\|_2^2
   \right].
\end{align}
The final line is exactly the ERGM objective under the elastic-net penalty in \eqref{prior}. Therefore,
\[
\hat{\th}_{MAP}=\hat{\th}_{EN},
\]
and the elastic-net penalized estimator is the posterior mode.

\begin{theorem}[MAP equivalence]
\label{thm:map-equivalence}
Fix \(\lambda_1>0\) and \(\lambda_2>0\). Under the elastic-net prior in
\eqref{penalty}, the posterior mode satisfies
\[
\widehat{\th}_{\mathrm{MAP}}
=
\arg\max_{\th\in\mathbb R^p}
\left\{
\th^\top s(\bs y)-\log z(\th)
-\lambda_1\|\th\|_1-\lambda_2\|\th\|_2^2
\right\}.
\]
Thus the posterior mode is exactly the elastic-net penalized full-likelihood
ERGM estimator.
\end{theorem}

\begin{proof}
The statement follows by taking logarithms of
\(\pi(\bs y\mid\th)\pi(\th)\) under \eqref{penalty}. Terms that do not depend
on \(\th\), including the prior normalizing constant, do not affect the
maximizer.
\end{proof}

To obtain a computationally useful Bayesian model, we use a hierarchical
latent-scale representation, following the spirit of \citet{Park2008}. Define
\begin{align}
    \pi(\th)
    &\propto
    \exp\!\left(
    -\lambda_1\|\th\|_1
    -\lambda_2\|\th\|_2^2
    \right) \\
    &=
    \prod_{j=1}^p
    \exp\!\left(
    -\lambda_1|\theta_j|
    -\lambda_2\theta_j^2
    \right).
\end{align}

\begin{lemma}\label{lemma:normalizing_constant}
Let \(\lambda_1>0\) and \(\lambda_2>0\). For the one-dimensional density
\[
\pi(\theta_j\mid \lambda_1,\lambda_2)
=
C(\lambda_1,\lambda_2)
\exp\!\left(
-\lambda_1|\theta_j|
-\lambda_2\theta_j^2
\right),
\]
the normalizing constant is
\[
C(\lambda_1,\lambda_2)
=
\sqrt{\lambda_2}
\exp\left(
-\frac{\lambda_1^2}{4\lambda_2}
\right)
\Gamma_U^{-1}
\left(
\frac12,
\frac{\lambda_1^2}{4\lambda_2}
\right),
\]
where
\[
\Gamma_U(s,x)=\int_x^\infty t^{s-1}e^{-t}\,dt
\]
denotes the upper incomplete gamma function.
\end{lemma}

The proof of Lemma~\ref{lemma:normalizing_constant} is given in the Appendix.

\begin{theorem}[Latent representation of the elastic-net prior]
\label{thm:latent-representation}
Let $\lambda_1>0$ and $\lambda_2>0$. The joint prior density of
\(\th=(\theta_1,\dots,\theta_p)^\top\) admits the representation
\begin{equation}\label{lemma1}
\resizebox{\linewidth}{!}{$
\pi(\th\mid \lambda_1,\lambda_2)
=
\left[
\frac{\lambda_1}
{2\sqrt{\pi}\,
\Gamma_U\left(\frac12,\frac{\lambda_1^2}{4\lambda_2}\right)}
\right]^p
\prod_{j=1}^p
\int_1^\infty
(t_j-1)^{-1/2}
\exp\!\left(
-\frac{\lambda_2 t_j}{t_j-1}\theta_j^2
-\frac{\lambda_1^2}{4\lambda_2}t_j
\right)\,dt_j .
$}
\end{equation}
\end{theorem}

\begin{proof}
Following \cite{andrews1974scale}, the Laplace density can be expressed as a
scale mixture of normal distributions:
\begin{equation}\label{laplace_andrew}
\frac{a}{2}e^{-a|z|}
=
\int_0^\infty
\frac{1}{\sqrt{2\pi s}}
\exp\!\left(-\frac{z^2}{2s}\right)
\left(
\frac{a^2}{2}e^{-a^2s/2}
\right)ds,
\qquad a>0.
\end{equation}
Taking \(a=\lambda_1\) and \(z=\theta_j\), we obtain
\begin{align}
e^{-\lambda_1|\theta_j|}
&=
\lambda_1
\int_0^\infty
\frac{1}{\sqrt{2\pi s}}
\exp\!\left[
-\frac12
\left(
\frac{\theta_j^2}{s}
+
\lambda_1^2s
\right)
\right]ds .
\end{align}
Multiplying both sides by \(\exp\{-\lambda_2\theta_j^2\}\) gives
\begin{align}
e^{-\lambda_1|\theta_j|-\lambda_2\theta_j^2}
&=
\lambda_1
\int_0^\infty
\frac{1}{\sqrt{2\pi s}}
\exp\!\left[
-\frac{\lambda_1^2}{2}s
-
\left(
\frac{1}{2s}
+
\lambda_2
\right)\theta_j^2
\right]ds .
\end{align}
Now let
\[
t_j=1+2\lambda_2s,
\qquad
s=\frac{t_j-1}{2\lambda_2},
\qquad
ds=\frac{1}{2\lambda_2}\,dt_j .
\]
Then
\[
\frac{1}{2s}+\lambda_2
=
\frac{\lambda_2t_j}{t_j-1},
\qquad
\frac{\lambda_1^2}{2}s
=
\frac{\lambda_1^2}{4\lambda_2}(t_j-1).
\]
Therefore,
\begin{align}
e^{-\lambda_1|\theta_j|-\lambda_2\theta_j^2}
&=
\frac{\lambda_1}{2\sqrt{\pi\lambda_2}}
\int_1^\infty
(t_j-1)^{-1/2}
\exp\!\left(
-\frac{\lambda_2t_j}{t_j-1}\theta_j^2
-\frac{\lambda_1^2}{4\lambda_2}(t_j-1)
\right)\,dt_j .
\end{align}
Multiplying both sides by the normalizing constant
\[
C(\lambda_1,\lambda_2)
=
\sqrt{\lambda_2}
\exp\!\left(
-\frac{\lambda_1^2}{4\lambda_2}
\right)
\Gamma_U^{-1}
\left(
\frac12,
\frac{\lambda_1^2}{4\lambda_2}
\right),
\]
we get
\begin{align}
\pi(\theta_j\mid \lambda_1,\lambda_2)
&=
\frac{\lambda_1}
{2\sqrt{\pi}\,
\Gamma_U\left(\frac12,\frac{\lambda_1^2}{4\lambda_2}\right)}
\int_1^\infty
(t_j-1)^{-1/2} \times
\exp\!\left(
-\frac{\lambda_2t_j}{t_j-1}\theta_j^2
-\frac{\lambda_1^2}{4\lambda_2}t_j
\right)\,dt_j .
\end{align}
Since the coordinates \(\theta_1,\dots,\theta_p\) are conditionally independent
under the prior, taking the product over \(j=1,\dots,p\) yields
\[
\pi(\th\mid \lambda_1,\lambda_2)
=
\prod_{j=1}^p
\pi(\theta_j\mid \lambda_1,\lambda_2),
\]
which gives the representation in \eqref{lemma1}. This completes the proof.
\end{proof}

Theorem~\ref{thm:latent-representation} yields the following hierarchical model:
\begin{align}
    \pi(y \mid \th)
    &=
    \frac{\exp\{\th^\top s(y)\}}{z(\th)}, \\
    \pi(\theta_j \mid t_j,\lambda_2)
    &=
    \sqrt{\frac{t_j\lambda_2}{\pi(t_j-1)}}
    \exp\!\left(
    -\frac{t_j\lambda_2}{t_j-1}\theta_j^2
    \right),
    \label{eq:theta_given_t}\\
    \pi(t_j\mid \lambda_1,\lambda_2)
    &=
    \frac{\delta^{1/2}}{\Gamma_U(1/2,\delta)}
    t_j^{-1/2}\exp(-\delta t_j)\mathbb I(t_j>1),
    \qquad
    \delta=\frac{\lambda_1^2}{4\lambda_2}.
    \label{eq:t_prior}
\end{align}
Thus \(t_j\mid\lambda_1,\lambda_2\) is a Gamma distribution with shape \(1/2\) and rate \(\delta\), truncated to \((1,\infty)\), and
\[
\theta_j\mid t_j,\lambda_2
\sim
N\!\left(0,\frac{t_j-1}{2t_j\lambda_2}\right).
\]

For posterior simulation it is more convenient to update \(u_j=t_j-1\). Combining \eqref{eq:theta_given_t} and \eqref{eq:t_prior} gives
\[
\pi(u_j\mid \theta_j,\lambda_1,\lambda_2)
\propto
u_j^{-1/2}
\exp\!\left\{
-\frac12\left(
2\lambda_2\theta_j^2u_j^{-1}
+
\frac{\lambda_1^2}{2\lambda_2}u_j
\right)
\right\},
\qquad u_j>0.
\]
Hence
\begin{equation}\label{eq:gig_update}
u_j\mid \theta_j,\lambda_1,\lambda_2
\sim
\operatorname{GIG}
\left(
\frac12,
2\lambda_2\theta_j^2,
\frac{\lambda_1^2}{2\lambda_2}
\right),
\end{equation}
where \(\operatorname{GIG}(\lambda,\chi,\psi)\) has density \citep{jorgensen2012statistical}
\[
f(x\mid\lambda,\chi,\psi)
=
\frac{(\psi/\chi)^{\lambda/2}}
{2K_\lambda(\sqrt{\psi\chi})}
x^{\lambda-1}
\exp\!\left\{-\frac12(\chi x^{-1}+\psi x)\right\},
\qquad x>0.
\]

For every \(\lambda_1,\lambda_2>0\), Lemma~\ref{lemma:normalizing_constant}
gives a proper coordinate density, and hence a proper product prior for
\(\th\). Because the network sample space is finite, \(\pi(y\mid\th)\le 1\).
The marginal likelihood is therefore bounded above by the prior integral, so
the posterior is proper, see the proof in the Appendix \ref{app:posterior propriety}.

\subsection{Choosing \texorpdfstring{$\lambda_1$ and $\lambda_2$}{lambda1 and lambda2} by empirical Bayes}

Following \citet{li2010bayesian} and \citet{casella2001empirical}, we estimate the
elastic-net hyperparameters \((\lambda_1,\lambda_2)\) by a Monte Carlo EM (MCEM)
algorithm. Under the hierarchical representation above, we treat
\[
(\th,\bs t), \qquad \bs t=(t_1,\dots,t_p)^\top,
\]
as missing data, while \((\lambda_1,\lambda_2)\) are regarded as fixed hyperparameters.
Let
\[
\delta=\delta(\lambda_1,\lambda_2)
=\frac{\lambda_1^2}{4\lambda_2}.
\]
Then the complete-data joint density is
\begin{align}
\pi(\bs y,\th,\bs t\mid \lambda_1,\lambda_2)
&=
\pi(\bs y\mid \th)
\prod_{j=1}^p
\pi(\theta_j\mid t_j,\lambda_2)\,
\pi(t_j\mid \lambda_1,\lambda_2)
\nonumber\\
&\propto
\frac{\exp\{\th^\top s(\bs y)\}}{z(\th)}
\left[
\frac{\lambda_1}
{2\sqrt{\pi}\,\Gamma_U(1/2,\delta)}
\right]^p
\prod_{j=1}^p
(t_j-1)^{-1/2}
\notag\\
&\qquad\times
\exp\!\left(
-\frac{t_j\lambda_2}{t_j-1}\theta_j^2
-\delta t_j
\right)
\mathbb I(t_j>1).
\label{eq:complete_data_density_lambda_main}
\end{align}

At the \(k\)th MCEM iteration, suppose that
\[
\bigl\{(\th^{(m)},\bs t^{(m)})\bigr\}_{m=1}^M
\sim
\pi(\th,\bs t\mid \bs y,\lambda^{(k-1)}),
\qquad
\lambda^{(k-1)}=\bigl(\lambda_1^{(k-1)},\lambda_2^{(k-1)}\bigr).
\]
Define the Monte Carlo summaries
\begin{align}
A_M^{(k-1)}
&=
\frac{1}{M}\sum_{m=1}^M
\sum_{j=1}^p
\frac{t_j^{(m)}}{t_j^{(m)}-1}\bigl(\theta_j^{(m)}\bigr)^2,
\label{eq:AM_main}
\\
B_M^{(k-1)}
&=
\frac{1}{M}\sum_{m=1}^M
\sum_{j=1}^p t_j^{(m)}.
\label{eq:BM_main}
\end{align}
Then the M-step update for \(\lambda_2\) is available in closed form:
\begin{equation}
\lambda_2^{(k)}
=
\frac{p}{2A_M^{(k-1)}}.
\label{eq:lambda2_update_main}
\end{equation}

Next, \(\lambda_1^{(k)}\) is defined as the positive solution of the one-dimensional
equation
\begin{equation}
g_k(\lambda_1)=0,
\label{eq:lambda1_root_main}
\end{equation}
where
\begin{align}
g_k(\lambda_1)
&=
\frac{p}{\lambda_1}
+\frac{p\lambda_1}{2\lambda_2^{(k)}}\,h_k(\lambda_1)
-\frac{\lambda_1}{2\lambda_2^{(k)}}\,B_M^{(k-1)},
\label{eq:gk_main}
\\
h_k(\lambda_1)
&=
\frac{\exp\{-\delta_k(\lambda_1)\}}
{\Gamma_U\!\left(\frac12,\delta_k(\lambda_1)\right)\sqrt{\delta_k(\lambda_1)}},
\label{eq:hm_main}
\\
\delta_k(\lambda_1)
&=
\frac{\lambda_1^2}{4\lambda_2^{(k)}}.
\label{eq:deltam_main}
\end{align}
Equivalently,
\begin{equation}
\lambda_1^2
=
\frac{2p\lambda_2^{(k)}}{
B_M^{(k-1)}
-p\,h_k(\lambda_1)
}.
\label{eq:lambda1_fixed_point_main}
\end{equation}

Therefore, each MCEM iteration consists of the following two steps:
(i) sample \((\th,\bs t)\) from the posterior distribution under the current
hyperparameter values and compute \(A_M^{(k-1)}\) and \(B_M^{(k-1)}\);
(ii) update \(\lambda_2^{(k)}\) by \eqref{eq:lambda2_update_main} and obtain
\(\lambda_1^{(k)}\) by numerically solving \eqref{eq:lambda1_root_main}. For numerical
stability, the latter step can be carried out on the log-scale, \(\eta=\log\lambda_1\).

\begin{proposition}[Unique global maximizer of the Monte Carlo M-step]
\label{prop:unique-mcem-mstep}
Assume \(A_M>0\) and \(B_M/p>1\). Under the reparameterization
\[
\delta=\frac{\lambda_1^2}{4\lambda_2},
\qquad
\lambda_1=2\sqrt{\delta\lambda_2},
\]
the Monte Carlo objective separates as
\begin{align}
\widehat Q_M(\delta,\lambda_2)
&=
\frac p2\log\lambda_2-\lambda_2A_M
+\frac p2\log\delta
-p\log\Gamma_U\!\left(\frac12,\delta\right)
-\delta B_M+C,
\label{eq:Qhat_delta_parameterization}
\end{align}
where \(C\) does not depend on \((\delta,\lambda_2)\). It has a unique global
maximizer given by
\[
\widehat\lambda_2=\frac{p}{2A_M},
\]
and by the unique positive solution \(\widehat\delta\) of
\begin{equation}
\frac{B_M}{p}
=
\frac{1}{2\delta}
+
\frac{e^{-\delta}}
{\sqrt{\delta}\,\Gamma_U(1/2,\delta)}.
\label{eq:delta_unique_root}
\end{equation}
The corresponding update is
\[
\widehat\lambda_1
=2\sqrt{\widehat\delta\,\widehat\lambda_2}.
\]
\end{proposition}

\begin{proof}
Equation~\eqref{eq:Qhat_delta_parameterization} follows from
\(p\log\lambda_1=p\log 2+(p/2)\log\delta+(p/2)\log\lambda_2\).
The \(\lambda_2\)-part is strictly concave and is uniquely maximized at
\(p/(2A_M)\). For the \(\delta\)-part, differentiation gives
\[
\frac{\partial\widehat Q_M}{\partial\delta}
=
\frac{p}{2\delta}
+p\frac{e^{-\delta}}
{\sqrt{\delta}\,\Gamma_U(1/2,\delta)}
-B_M.
\]
The function on the right-hand side of
\eqref{eq:delta_unique_root} is the mean of a Gamma random variable with shape
\(1/2\) and rate \(\delta\), truncated to \((1,\infty)\). Differentiating this
mean with respect to the rate gives minus its variance, so it is strictly
decreasing. It diverges as \(\delta\downarrow0\) and converges to \(1\) as
\(\delta\to\infty\). Since \(B_M/p>1\), a unique positive root exists. The
strict concavity of each separated component proves uniqueness of the global
maximizer.
\end{proof}

\section{Computational Algorithms}\label{sec:algorithms}
The sampler combines the exchange update for the ERGM likelihood with Gibbs
updates for the elastic-net latent scales. We first recall the approximate
exchange algorithm used in Bayesian ERGM computation \citep{caimo2017bayesian}.
The exchange step proposes a parameter value, simulates an auxiliary network at
the proposed value, and uses the auxiliary statistic to cancel the intractable
likelihood normalizing constants in the Metropolis--Hastings ratio. In
population implementations, parallel adaptive direction sampling is often used
to improve movement across chains.

The proposed elastic-net BERGM augments the exchange step with the latent-scale
updates from Section~\ref{sec:model}. The proposal covariance and the
differential-evolution weight \(\gamma\) are fixed inputs, while the
implementation updates \((\lambda_1,\lambda_2)\) from a rolling Monte Carlo
window whenever MCEM is enabled, See Algorithm \ref{alg:bergm_elastic_net}.

\begin{algorithm}[H]
\caption{Bayesian Elastic-Net BERGM with Empirical-Bayes Penalty Updates}
\label{alg:bergm_elastic_net}
\begin{algorithmic}[1]
    \State Choose initial \(\lambda_1,\lambda_2>0\), fixed \(\gamma\), fixed
    proposal covariance \(S\), chains \(H\), burn-in \(B\), main iterations
    \(R\), and the total iterations \(N=B+R\).
    \State Compute \(s(y)\). Obtain MPLE starting values if available; otherwise
    use user-supplied values or zero. Initialize each \(\th_h\) with a small
    random perturbation.
    \State Set \(t_{jh}=2\) and initialize the rolling MCEM buffer of length \(W\).
    \For{$i = 1,\dots,N$}
        \For{$h = 1,\dots,H$}
            \State Draw an ordered pair \((h_1,h_2)\) uniformly from chains with
            \(h_1\ne h_2\) and \(h_1,h_2\ne h\), and propose
            \(\th'_h=\th_h+\gamma(\th_{h_1}-\th_{h_2})+\epsilon\),
            \(\epsilon\sim N_p(0,S)\).
            \State Simulate \(y'\sim \pi(\cdot\mid\th'_h)\) by an auxiliary ERGM MCMC run.
            \State Set \(V_{jh}=(t_{jh}-1)/(2t_{jh}\lambda_2)\).
            \State Accept \(\th'_h\) with log probability
            \(\min\{0,(\th_h-\th'_h)^\top[s(y')-s(y)]
            +\log\phi_p(\th'_h;\mu,V_h)-\log\phi_p(\th_h;\mu,V_h)\}\).
            \For{$j=1,\dots,p$}
                \State Set \(\chi_{jh}=2\lambda_2(\theta_{jh}-\mu_j)^2\)
                and \(\psi_h=\lambda_1^2/(2\lambda_2)\).
                \State Draw \(u_{jh}\sim\operatorname{GIG}(1/2,\chi_{jh},\psi_h)\)
                and set \(t_{jh}=1+u_{jh}\).
            \EndFor
        \EndFor
        \State Store \((\th_h,t_h)_{h=1}^H\) in the rolling MCEM buffer.
        \If{the MCEM buffer is available and an MCEM update is scheduled}
            \State Compute \(A_M,B_M\); set \(\lambda_2=p/(2A_M)\); solve
            \eqref{eq:lambda1_root_main} for \(\lambda_1\).
        \EndIf
        \If{$i>B$}
            \State Save \((\th_h,t_h)_{h=1}^H\), current
            \((\lambda_1,\lambda_2)\), and acceptance indicators.
        \EndIf
    \EndFor
\end{algorithmic}
\end{algorithm}
\begin{theorem}[Invariance of the ideal fixed-penalty population kernel]
\label{thm:ideal-population-invariance}
Fix \(\lambda_1,\lambda_2>0\). Suppose that, when chain \(h\) is updated, the
ordered pair \((h_1,h_2)\) is sampled uniformly from all admissible ordered
pairs, the innovation distribution is symmetric about zero, and the auxiliary
network is sampled exactly from \(\pi(\cdot\mid\th'_h)\). Then the sequential
population exchange update followed by the GIG updates for \(\bs t_h\) leaves
\[
\Pi_H(d\th_{1:H},d\bs t_{1:H})
=
\prod_{h=1}^H
\pi(\th_h,\bs t_h\mid\bs y,\lambda_1,\lambda_2)
\,d\th_h\,d\bs t_h
\]
invariant.
\end{theorem}

\begin{proof}
Condition on all chains other than \(h\). Uniform sampling of ordered pairs and
symmetry of the Gaussian innovation imply that the proposal increment has the
same distribution as its negative. Thus
\(q_h(\th'_h\mid\th_h)=q_h(\th_h\mid\th'_h)\) conditional on the remaining
population. With exact \(\bs y'\sim\pi(\cdot\mid\th'_h)\), the exchange
acceptance probability satisfies detailed balance for the conditional target
\(\pi(\th_h\mid\bs t_h,\bs y,\lambda_1,\lambda_2)\); the intractable ERGM
normalizing constants cancel in the standard exchange construction. The
subsequent coordinatewise GIG draws are Gibbs updates from
\(\pi(\bs t_h\mid\th_h,\bs y,\lambda_1,\lambda_2)\). Therefore the complete
update of chain \(h\) preserves its augmented conditional posterior. Sequential
composition over \(h=1,\ldots,H\) preserves the product measure \(\Pi_H\).
\end{proof}

\begin{proposition}[Conditional perturbation bound for finite auxiliary simulation]
\label{prop:auxiliary-perturbation}
Let \(P\) denote the ideal fixed-penalty kernel in
Theorem~\ref{thm:ideal-population-invariance}, and let \(P_L\) denote the same
kernel when the auxiliary network is generated by \(L\) transitions of an ERGM
kernel \(K_{\th'}\). If
\[
\sup_{\th',\bs x}
\left\|K_{\th'}^L(\bs x,\cdot)-\pi(\cdot\mid\th')\right\|_{\mathrm{TV}}
\le \varepsilon_L,
\]
then
\[
\sup_{z}\|P_L(z,\cdot)-P(z,\cdot)\|_{\mathrm{TV}}
\le \varepsilon_L.
\]
If, in addition, the Dobrushin contraction coefficient of \(P\) satisfies
\(\tau(P)<1\), and \(P_L\) has invariant distribution \(\Pi_L\), then
\[
\|\Pi_L-\Pi_H\|_{\mathrm{TV}}
\le
\frac{\varepsilon_L}{1-\tau(P)}.
\]
\end{proposition}

\begin{proof}
For fixed current and proposed parameter values, the exchange acceptance
function is bounded between zero and one. Replacing the exact auxiliary law by
an auxiliary law within \(\varepsilon_L\) in total variation therefore changes
the expected acceptance probability by at most \(\varepsilon_L\). Integrating
over the proposal and including the complementary rejection probability gives
the first bound. For the second,
\[
\|\Pi_L-\Pi_H\|_{\mathrm{TV}}
\le
\|\Pi_LP_L-\Pi_LP\|_{\mathrm{TV}}
+
\|\Pi_LP-\Pi_HP\|_{\mathrm{TV}}
\le
\varepsilon_L+\tau(P)\|\Pi_L-\Pi_H\|_{\mathrm{TV}},
\]
and rearrangement gives the result.
\end{proof}

\section{Thresholded Reporting and Grouping Interpretation}
\label{sec:theory}

This section explains how term reporting and coefficient grouping are
interpreted in the present framework. Because the elastic-net prior is
continuous, the posterior is absolutely continuous and therefore
\[
\Pi(\theta_j=0\mid\bs y)=0.
\]
Thus, although the \(L_1\) penalty may place a MAP coefficient exactly at zero,
the model does not perform Bayesian variable selection in the spike-and-slab
sense. We instead use ``active'' as a post-fitting reporting label. Likewise,
the grouping results concern closeness of MAP coefficients rather than the
posterior probability that two terms are jointly selected. The main text states
the definitions and results; all proofs and supporting calculations are
collected in Appendix~\ref{app:reporting-grouping-proofs}.

\subsection{Decision-theoretic reporting rule}

For a fixed-hyperparameter posterior, define the practical-activity event
\[
A_j(\delta)=\{|\theta_j|>\delta\},
\qquad
p_j(\delta)=\Pi\{A_j(\delta)\mid\bs y\},
\]
where \(\delta>0\) is a prespecified threshold on the coefficient scale. Let
\(d_j=1\) indicate that term \(j\) is reported as active and \(d_j=0\) that it
is not reported. Consider the loss
\[
L(d_j,\theta_j)
=
c_{\mathrm{FP}}\mathbb I\{d_j=1,A_j(\delta)^c\}
+
c_{\mathrm{FN}}\mathbb I\{d_j=0,A_j(\delta)\},
\]
where \(c_{\mathrm{FP}},c_{\mathrm{FN}}>0\) are the losses assigned to false
positive and false negative decisions, respectively.

\begin{proposition}[Bayes reporting rule]
\label{prop:reporting-decision}
The posterior Bayes action reports term \(j\) as active when
\[
p_j(\delta)
>
\frac{c_{\mathrm{FP}}}
     {c_{\mathrm{FP}}+c_{\mathrm{FN}}}.
\]
Consequently, the cutoff \(0.90\) used in the standardized simulation
corresponds to
\[
c_{\mathrm{FP}}=9c_{\mathrm{FN}}.
\]
\end{proposition}

The proof of Proposition~\ref{prop:reporting-decision} is given in
Appendix~\ref{app:proof-reporting}.

In the adaptive implementation, \(p_j(\delta)\) is approximated using \(B\)
main draws:
\[
\widehat p_j(\delta)
=
\frac{1}{B}\sum_{b=1}^{B}
\mathbb I\{|\theta_j^{(b)}|>\delta\}.
\]
We report term \(j\) as active when
\(\widehat p_j(\delta)>0.90\). Because the penalties are adapted during the run
and auxiliary networks are generated by finite ERGM MCMC,
\(\widehat p_j(\delta)\) is an approximate empirical-Bayes reporting summary,
not an exact posterior inclusion probability.

\subsection{Coefficient-grouping results}

The elastic-net penalty can also encourage similar coefficients for terms that
represent similar features. The appropriate notion of similarity depends on
whether one considers the logistic pseudolikelihood or the full ERGM
likelihood.

\paragraph{Pseudolikelihood.}
Let \(m=1,\ldots,M\) index the observed dyads, let
\(x_m=\Delta s_m(\bs y)\) be the corresponding change-statistic vector, and
write
\[
X_j=(x_{1j},\ldots,x_{Mj})^\top
\]
for its \(j\)th design column. Let
\(\widehat\th^{\mathrm{PL}}\) denote the elastic-net penalized
pseudolikelihood estimator defined in Section~\ref{sec:model}.

\begin{proposition}[Pseudolikelihood coefficient grouping]
\label{prop:pl-grouping}
Suppose that \(\lambda_2>0\). If
\(\widehat\theta_j^{\mathrm{PL}}\) and
\(\widehat\theta_k^{\mathrm{PL}}\) are nonzero and have the same sign, then
\[
\left|
\widehat\theta_j^{\mathrm{PL}}
-
\widehat\theta_k^{\mathrm{PL}}
\right|
\le
\frac{\|X_j-X_k\|_1}{2\lambda_2}
\le
\frac{\sqrt M\,\|X_j-X_k\|_2}{2\lambda_2}.
\]
In particular, under these conditions, identical change-statistic columns
yield identical coefficients.
\end{proposition}

The KKT derivation and the corresponding standardized-design calculation are
provided in Appendix~\ref{app:pl-grouping-proof}. This result is a property of
the logistic pseudolikelihood design. It is useful for screening and
initialization but does not describe the full-likelihood posterior targeted by
the sampler.

\paragraph{Full likelihood.}
For the full ERGM likelihood, there is no fixed observed design matrix.
Similarity must instead be measured by how the sufficient statistics differ
over the network sample space. Let \(\widehat\th\) denote the full-likelihood
elastic-net MAP estimator defined in Section~\ref{sec:model}. For two
sufficient statistics \(s_j\) and \(s_k\), define
\[
g_{jk}(\bs u)=s_j(\bs u)-s_k(\bs u)
\]
and its oscillation over \(\mathcal Y\) by
\[
D_{jk}
=
\sup_{\bs u,\bs v\in\mathcal Y}
\left|
g_{jk}(\bs u)-g_{jk}(\bs v)
\right|.
\]

\begin{theorem}[Full-likelihood MAP coefficient grouping]
\label{thm:full-likelihood-grouping-effect}
Suppose that \(\mathcal Y\) is finite,
\(\lambda_1\geq0\), and \(\lambda_2>0\). If
\[
\widehat\theta_j\neq0,\qquad
\widehat\theta_k\neq0,\qquad
\operatorname{sign}(\widehat\theta_j)
=
\operatorname{sign}(\widehat\theta_k),
\]
then
\[
2\lambda_2(\widehat\theta_j-\widehat\theta_k)
=
g_{jk}(\bs y)
-
\mathbb E_{\widehat\th}\{g_{jk}(\bs Y)\},
\]
and consequently
\[
|\widehat\theta_j-\widehat\theta_k|
\le
\frac{D_{jk}}{2\lambda_2}.
\]
In particular, if \(s_j-s_k\) is constant on \(\mathcal Y\), then
\[
\widehat\theta_j=\widehat\theta_k.
\]
\end{theorem}

The proof of
Theorem~\ref{thm:full-likelihood-grouping-effect} is given in
Appendix~\ref{app:full-likelihood-grouping-proof}.

\section{Simulation}
\label{sec:simulation}

This section studies the performance of BERGM Elastic Net when the fitted ERGM
is over-specified and some active covariates are strongly correlated. We
compare the proposed method with Standard BERGM \citep{CaimoFriel2014}, BERGM
Lasso-EB \citep{Han2024adaptiveLasso}, BERGM Ridge Cauchy-EB
\citep{modisette2023penalized}, and Horseshoe BERGM-EB
\citep{linares2026bayesian}. Performance is assessed through estimation error,
thresholded reporting, joint recovery and coefficient separation for the
correlated active pair, empirical interval behavior, posterior-predictive fit,
and computational failures.

The simulation is intended to assess whether elastic-net regularization can
shrink inactive terms while achieving stable estimates of correlated active
effects. Because the five methods use different priors, proposal mechanisms,
tuning rules, and adaptive updates, the results compare the complete
implementations rather than the isolated effect of the elastic-net prior. The
following subsections describe the simulation design, evaluation criteria, and
results.

\subsection{Simulation Design}

We use \(R=50\) Monte Carlo replications, with base seed \(20260511\). In each
replication, an undirected, binary, loopless network is generated on \(n=50\)
nodes. The nodes are assigned to three approximately balanced groups. Two
continuous covariates, \(x_1\) and \(x_2\), are generated from a bivariate
normal distribution with mean zero, marginal variances one, and correlation $
\operatorname{corr}(x_1,x_2)=0.95.$
Five additional covariates, \(z_1,\ldots,z_5\), are generated independently
from standard normal distributions and are included as inactive noise
covariates.

The data-generating ERGM is
\[
\begin{aligned}
Y\sim\;&
\texttt{edges}
+\texttt{nodematch("group")}
+\texttt{nodecov("x1")}
+\texttt{nodecov("x2")}\\
&+\texttt{gwesp}(0.5,\texttt{fixed=TRUE}).
\end{aligned}
\]
The non-edge coefficients are
\[
\begin{aligned}
\theta_{\texttt{nodematch(group)}}&=0.35, &
\theta_{\texttt{nodecov(x1)}}&=0.65,\\
\theta_{\texttt{nodecov(x2)}}&=0.65, &
\theta_{\texttt{gwesp}}&=0.35.
\end{aligned}
\]

The edge coefficient is calibrated to produce an expected network density near
\(0.07\). We consider a grid over \([-8,-1]\) with spacing \(0.25\) and
simulate \(100\) networks at each grid value. Networks with density below
\(0.005\) or above \(0.50\) are classified as degenerate during calibration.
This procedure selects an edge coefficient of \(-3.75\), for which the mean
simulated density is \(0.0717\) to avoid the degeneracy issue.

Each simulated network is fitted with the same over-specified candidate model:
\[
\begin{aligned}
Y\sim\;&
\texttt{edges}+\texttt{nodematch("group")}
+\texttt{nodecov("x1")}+\texttt{nodecov("x2")}\\
&+\texttt{absdiff("x1")}+\texttt{absdiff("x2")}
+\texttt{gwesp}(0.5,\texttt{fixed=TRUE})\\
&+\texttt{gwdsp}(0.5,\texttt{fixed=TRUE})
+\texttt{triangle}+\texttt{kstar(2)}+\texttt{kstar(3)}\\
&+\sum_{k=1}^{5}\texttt{nodecov}(z_k).
\end{aligned}
\]
The candidate model contains \(p=16\) coefficients: the edge term, four active
non-edge terms, and eleven inactive non-edge terms. The active non-edge terms
are group matching, the two correlated nodal covariates \(x_1\) and \(x_2\),
and the geometrically weighted edgewise shared-partner statistic. The edge
coefficient is treated as a density parameter and is excluded from the
estimation and reporting metrics.

All methods use finite MPLE starting values when these are available, with zero
vectors used only as a fallback. Each BERGM fit uses four chains, \(100\)
burn-in iterations, \(2000\) main iterations, \(500\) auxiliary ERGM
iterations, and no additional thinning. For the four comparison methods, the
differential-evolution weight is \(\gamma=0.5\), and the scalar proposal
variance is \(0.0025\).

For BERGM Elastic Net, we use \(\gamma=0.05\) and initialize the penalty
parameters at the initial $
\lambda_1=\lambda_2=5.
$
The proposal covariance is based on the MPLE covariance matrix scaled by
\(0.10\), with its eigenvalues truncated to $
[10^{-8},0.02].
$ For each fitted model, \(100\) networks are generated for
posterior-predictive assessment. Bayesian goodness-of-fit plots for the first
replication are based on \(30\) draws, with \(500\) auxiliary
iterations for each draw. A non-edge term is reported as active when
$
\widehat p_j(0.05)>0.90.
$
\subsection{Evaluation Metrics}

Let \(\mathcal A\) and \(\mathcal N\) denote the sets of active and inactive
non-edge terms, respectively. The edge coefficient is excluded because it is
treated as a density parameter rather than a candidate term for reporting.
For replication \(r\), and coefficient \(j\), let
\(\widehat\theta_{rj}\) denote the posterior mean and let
\(\theta_{0j}\) denote the true coefficient. For any coefficient set
\(\mathcal S\subseteq\mathcal A\cup\mathcal N\), define
\[
\operatorname{MSE}(\mathcal S)
=
\frac{1}{R|\mathcal S|}
\sum_{r=1}^{R}\sum_{j\in\mathcal S}
\left(\widehat\theta_{rj}-\theta_{0j}\right)^2.
\]
We report this quantity for the active terms
\(\mathcal S=\mathcal A\), the inactive terms
\(\mathcal S=\mathcal N\), and all non-edge terms
\(\mathcal S=\mathcal A\cup\mathcal N\). Smaller MSE indicates more accurate
point estimation.

A non-edge term is reported as active in replication \(r\) when
\[
\widehat p_{rj}(0.05)>0.90.
\]
Let
\[
\widehat{\mathcal A}_{r}
=
\left\{
j\in\mathcal A\cup\mathcal N:
\widehat p_{rj}(0.05)>0.90
\right\}
\]
be the set of terms reported by method \(m\). Define
\[
\operatorname{TP}_{r}
=
|\widehat{\mathcal A}_{r}\cap\mathcal A|,
\qquad
\operatorname{FP}_{r}
=
|\widehat{\mathcal A}_{r}\cap\mathcal N|.
\]
The true-positive rate, false-positive rate, and false-discovery rate are
\[
\operatorname{TPR}_{r}
=
\frac{\operatorname{TP}_{r}}{|\mathcal A|},
\qquad
\operatorname{FPR}_{r}
=
\frac{\operatorname{FP}_{r}}{|\mathcal N|},
\qquad
\operatorname{FDR}_{r}
=
\frac{\operatorname{FP}_{r}}
{\max\{|\widehat{\mathcal A}_{r}|,1\}}.
\]
Thus, TPR is the proportion of truly active terms that are reported, FPR is
the proportion of inactive terms that are incorrectly reported, and FDR is
the proportion of reported terms that are inactive. A larger TPR and smaller
FPR and FDR are desirable, although these quantities necessarily depend on
the reporting threshold.

Empirical interval coverage is the proportion of true non-edge coefficients
contained in their \(95\%\) confidence intervals, and interval
length is the corresponding average interval width. Because the sampler
adapts the penalties during the main run, these quantities describe the
empirical behavior of the adaptive summaries; they are not asserted to be
exact fixed-target Bayesian coverage probabilities.

Grouped recovery of the two correlated active covariates is assessed by the
proportions of replications in which both or exactly one is reported, together
with their mean separation,
\[
\left|
\widehat\theta_{\texttt{nodecov(x1)}}
-
\widehat\theta_{\texttt{nodecov(x2)}}
\right|.
\]
Posterior-predictive stability is evaluated using degeneracy rates, simulated
network densities, and Bayesian goodness-of-fit summaries. All reported table
entries are Monte Carlo averages across \(R=50\) replications, with Monte Carlo
standard errors (MCSEs) in parentheses.
\subsection{Simulation Results}

The simulation was designed to examine two potential benefits of elastic-net
regularization in an over-specified ERGM: reducing the effect of inactive terms
and stabilizing the estimates of correlated active terms. The results provide
evidence for both features. BERGM Elastic Net has the smallest estimation
error and the strongest false-positive control, while also giving the most
balanced estimates for the correlated active covariates.

Table~\ref{tab:case1-estimation} reports estimation error, thresholded-reporting
performance, and empirical interval behavior. BERGM Elastic Net has the
smallest active, inactive, and overall non-edge MSE among the five methods. Its
active-term MSE is \(0.105\), compared with \(0.263\)--\(0.309\) for the other
methods. Its inactive-term MSE is \(0.035\), less than one-quarter of the
\(0.163\)--\(0.175\) range obtained by the alternatives. As a result, its
overall non-edge MSE is \(0.054\), compared with \(0.198\)--\(0.202\).
Figure~\ref{fig:case1-mse} shows the same pattern. The reduction in MSE is
therefore not limited to stronger shrinkage of zero coefficients; it is also
present for the active terms.

The thresholded-reporting results show the benefit of stronger noise
suppression, together with its cost. Under the rule
\(\widehat p_j(0.05)>0.90\), BERGM Elastic Net has an FPR of \(0.184\) and an
FDR of \(0.423\), both the smallest among the five methods. Thus, it reports
fewer inactive terms and produces a smaller proportion of false discoveries.
Its TPR is \(0.610\), however, compared with \(0.705\)--\(0.800\) for the
other methods. The method is therefore more conservative under the chosen
cutoff: it improves false-positive control but fails to report some active
terms.

BERGM Elastic Net also has empirical interval coverage \(0.965\), close to the
nominal \(0.95\) level, with an average interval length of \(1.292\). BERGM
Lasso-EB produces longer intervals and lower coverage, while the shorter
intervals from BERGM Ridge Cauchy-EB, Horseshoe BERGM-EB, and Standard BERGM
have substantially lower coverage.

\begin{table}[H]
\centering
\caption{Estimation, thresholded-reporting, and empirical interval results
across \(50\) replications. MSE is mean squared error; TPR, FPR, and FDR are
the true-positive, false-positive, and false-discovery rates. Empirical
coverage and interval length refer to the approximate \(95\%\) intervals from
the adaptive output. Entries are replication averages, with Monte
Carlo standard errors in parentheses. The edge coefficient is excluded.}
\label{tab:case1-estimation}
\resizebox{\textwidth}{!}{%
\begin{tabular}{lcccccccc}
\toprule
Method & MSE active & MSE inactive & MSE overall & TPR & FPR & FDR &
Empirical coverage & Approx. interval length \\
\midrule
BERGM Elastic Net
& 0.105 (0.011) & 0.035 (0.003) & 0.054 (0.004)
& 0.610 (0.024) & 0.184 (0.014) & 0.423 (0.025)
& 0.965 (0.012) & 1.292 (0.021) \\
BERGM Lasso-EB
& 0.309 (0.052) & 0.163 (0.033) & 0.202 (0.032)
& 0.715 (0.029) & 0.307 (0.019) & 0.524 (0.021)
& 0.850 (0.027) & 1.700 (0.051) \\
BERGM Ridge Cauchy-EB
& 0.279 (0.052) & 0.170 (0.029) & 0.199 (0.029)
& 0.750 (0.027) & 0.425 (0.021) & 0.597 (0.014)
& 0.490 (0.035) & 0.606 (0.026) \\
Horseshoe BERGM-EB
& 0.277 (0.048) & 0.175 (0.031) & 0.202 (0.030)
& 0.705 (0.028) & 0.405 (0.020) & 0.600 (0.017)
& 0.410 (0.038) & 0.544 (0.023) \\
Standard BERGM
& 0.263 (0.046) & 0.174 (0.028) & 0.198 (0.027)
& 0.800 (0.028) & 0.418 (0.025) & 0.566 (0.021)
& 0.585 (0.038) & 0.747 (0.023) \\
\bottomrule
\end{tabular}
}
\end{table}

\begin{figure}[H]
    \centering
    \includegraphics[width=0.82\textwidth]
    {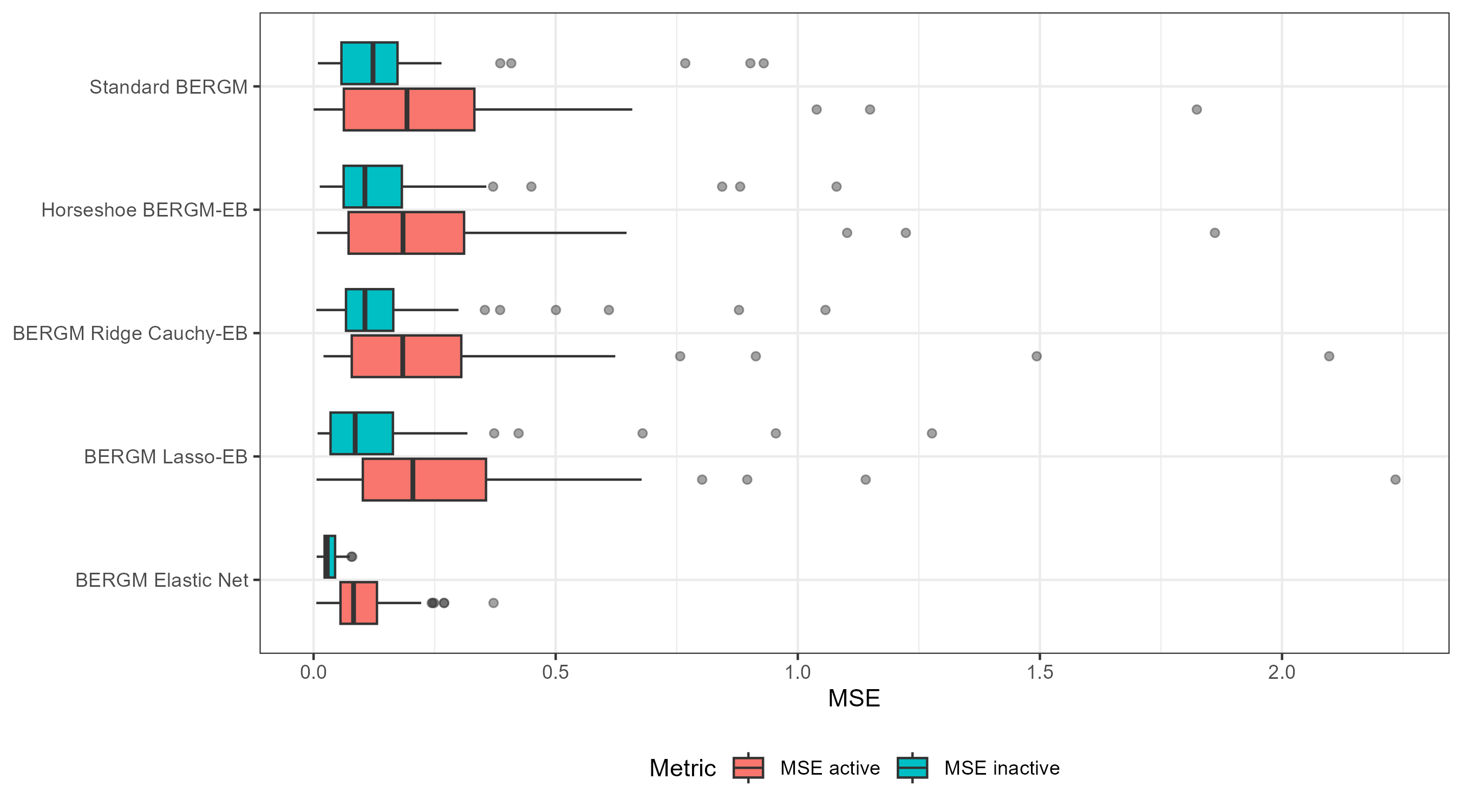}
    \caption{Active, inactive, and overall non-edge MSE across \(50\)
    replications. BERGM Elastic Net has the smallest MSE in all three
    categories.}
    \label{fig:case1-mse}
\end{figure}

MSE and thresholded reporting measure different features of the estimates. An
active coefficient may be estimated accurately even when its 
frequency does not exceed the reporting cutoff. This explains why BERGM
Elastic Net can have the smallest active-term MSE while also having a lower
TPR. Its low inactive-term MSE, FPR, and FDR reflect stronger control of
inactive terms, whereas the lower TPR reflects the cost of the conservative
reporting rule.

The second intended benefit of the elastic net is more stable estimation of
correlated signals. Table~\ref{tab:case1-grouping} examines this property for \(\texttt{nodecov(x1)}\) and \(\texttt{nodecov(x2)}\). BERGM Elastic Net
reports both active covariates in \(84\%\) of the replications and exactly one
in \(16\%\). BERGM Lasso-EB, BERGM Ridge Cauchy-EB, and Horseshoe BERGM-EB
each report both covariates in \(68\%\) of the replications, while Standard
BERGM does so in \(80\%\).

BERGM Elastic Net also gives the smallest coefficient separation,
\(0.522\), compared with \(0.850\)--\(0.991\) for the other methods. The higher paired-reporting rate and smaller coefficient separation indicate that the proposed method distributes the signal more evenly across the two correlated
covariates. This is the empirical behavior expected from combining the \(L_1\) and \(L_2\) components: the \(L_1\) component suppresses inactive
terms, while the \(L_2\) component discourages large differences between
related active coefficients.

\begin{table}[H]
\centering
\caption{Paired reporting and coefficient separation for the correlated active
covariates \(\texttt{nodecov(x1)}\) and \(\texttt{nodecov(x2)}\).
Coefficient separation is the absolute difference between their 
means. Entries are averages across \(50\) replications, with Monte Carlo
standard errors in parentheses.}
\label{tab:case1-grouping}
\resizebox{\textwidth}{!}{%
\begin{tabular}{lccc}
\toprule
Method & Both reported active & Exactly one reported active &
Coefficient separation \\
\midrule
BERGM Elastic Net
& 0.840 (0.052) & 0.160 (0.052) & 0.522 (0.055) \\
BERGM Lasso-EB
& 0.680 (0.067) & 0.320 (0.067) & 0.991 (0.105) \\
BERGM Ridge Cauchy-EB
& 0.680 (0.067) & 0.320 (0.067) & 0.890 (0.114) \\
Horseshoe BERGM-EB
& 0.680 (0.067) & 0.320 (0.067) & 0.898 (0.110) \\
Standard BERGM
& 0.800 (0.057) & 0.200 (0.057) & 0.850 (0.107) \\
\bottomrule
\end{tabular}
}
\end{table}

We also examined whether the additional shrinkage affects
posterior-predictive fit. Table~\ref{tab:case1-bgof} reports goodness-of-fit
diagnostics for the first simulation replication, based on degree,
geodesic-distance, edgewise shared-partner, and triad-census statistics. The
mean proportion of observed statistics within the \(90\%\) predictive bands
ranges from \(0.987\) to \(1.000\), and the MAE and RMSE values are similar
across methods. For this replication, the reduction in estimation error under
BERGM Elastic Net is not accompanied by an evident deterioration in
posterior-predictive fit. Figure~\ref{fig:case1-bgof-en} gives the full
diagnostic display for the BERGM Elastic Net fit. Since these diagnostics use
only the first replication, they are descriptive rather than a general
comparison of predictive performance.

\begin{table}[H]
\centering
\caption{Bayesian goodness-of-fit diagnostics for the first simulation
replication. The proportion within the \(90\%\) predictive band, mean absolute
error (MAE), and root mean squared error (RMSE) are aggregated over degree,
geodesic-distance, edgewise shared-partner, and triad-census statistics.}
\label{tab:case1-bgof}
\begin{tabular}{lccc}
\toprule
Method & Mean within 90\% band & MAE & RMSE \\
\midrule
BERGM Elastic Net       & 0.987 & 0.0039 & 0.0125 \\
BERGM Lasso-EB          & 1.000 & 0.0033 & 0.0109 \\
BERGM Ridge Cauchy-EB   & 1.000 & 0.0034 & 0.0103 \\
Horseshoe BERGM-EB      & 0.993 & 0.0033 & 0.0115 \\
Standard BERGM          & 0.993 & 0.0037 & 0.0116 \\
\bottomrule
\end{tabular}
\end{table}

\begin{figure}[H]
    \centering
    \includegraphics[width=0.86\textwidth]
    {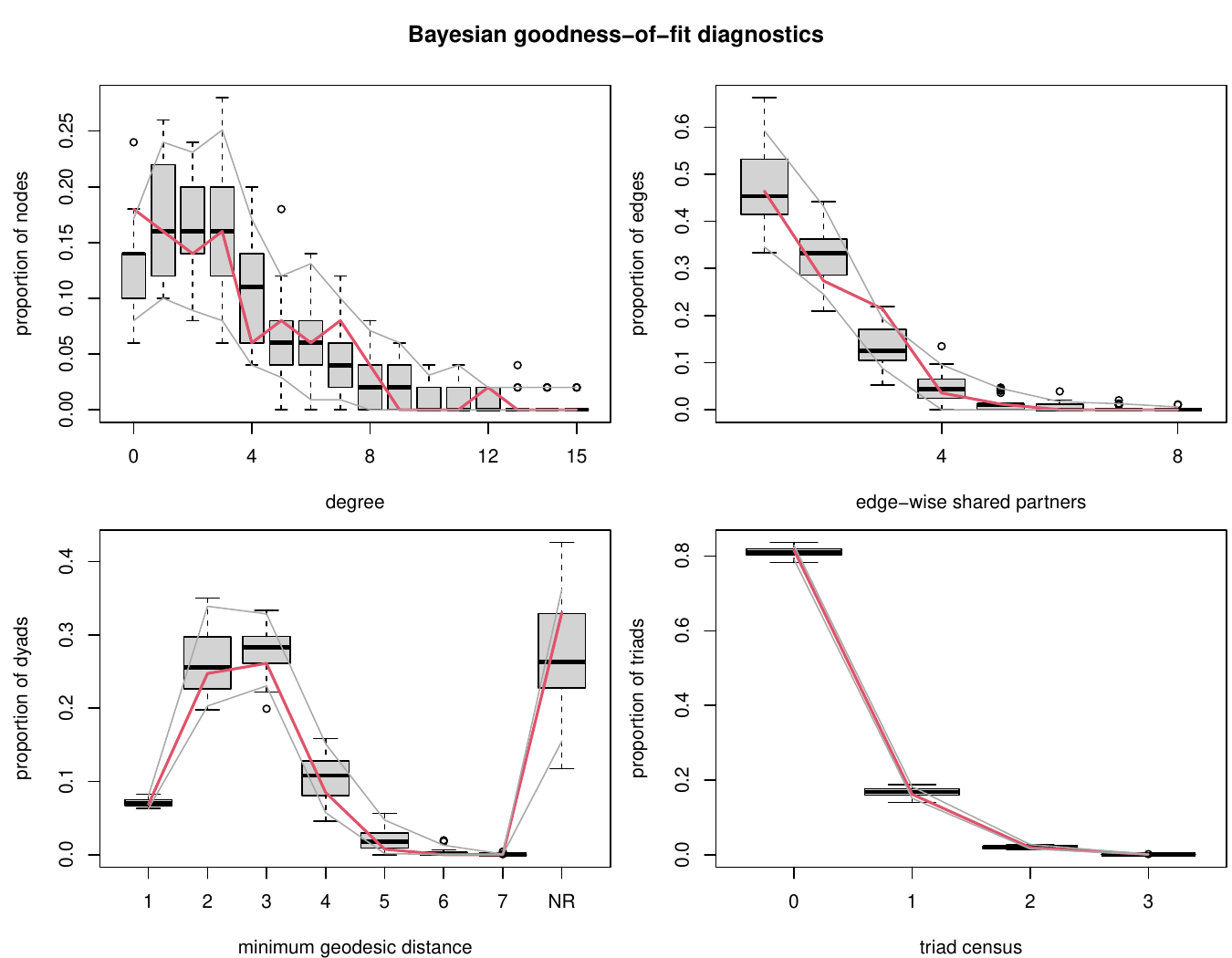}
    \caption{Bayesian goodness-of-fit diagnostics for the BERGM Elastic Net
    fit in the first simulation replication. The observed statistics are
    compared with their posterior-predictive distributions.}
    \label{fig:case1-bgof-en}
\end{figure}

In summary, the results illustrate the main reason for combining the two
penalty components. In this correlated and over-specified ERGM, BERGM Elastic
Net shrinks inactive terms more strongly than the comparison methods while
giving more similar estimates for the correlated active pair. It has the
smallest active, inactive, and overall MSE, the lowest FPR and FDR, the highest
paired-reporting rate, and the smallest coefficient separation. The cost is a
lower TPR under the \(0.90\) reporting cutoff. These findings support the use
of elastic-net regularization when an ERGM contains both correlated signals
and many inactive candidate terms.

\section{Applications}\label{sec:applications}

\subsection{\texttt{faux.magnolia.high} Friendship Network}
\label{subsec:fmh-elastic-net}

We first consider the \texttt{faux.magnolia.high} network, a benchmark school
friendship network distributed with the \texttt{ergm} package and based on
Wave~I of the National Longitudinal Study of Adolescent Health
\citep{goodreau2008statnet}. The network contains \(1{,}461\) students and
\(974\) undirected friendship ties. We use this example to examine two familiar
features of adolescent friendship networks: whether students in the same grade
are more likely to be friends and whether friendships exhibit transitive
closure. Figure~\ref{fig:clustering} displays the network with nodes colored by
grade.

\begin{figure}[H]
    \centering
    \includegraphics[width=0.45\linewidth]
    {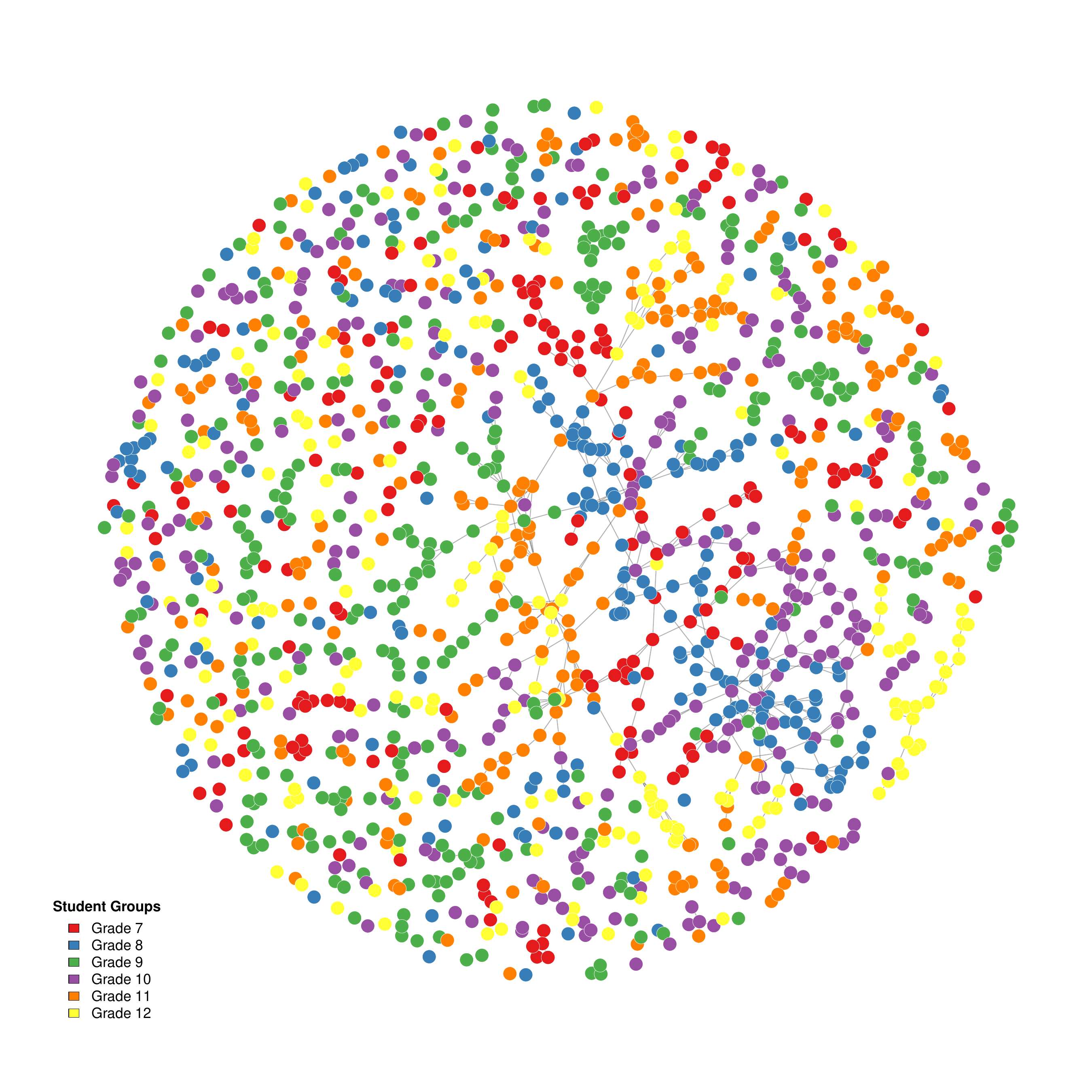}
    \caption{The \texttt{faux.magnolia.high} friendship network, with students
    colored by grade level.}
    \label{fig:clustering}
\end{figure}

Because the network is undirected, its statistics are defined over unordered
dyads. The edge statistic is
\[
\sum_{i<j}y_{ij},
\]
and the grade-matching statistic is
\[
\sum_{i<j}
y_{ij}\mathbb I(\text{Grade}_i=\text{Grade}_j).
\]
Reciprocity is not included because every undirected tie is reciprocal by
definition. Transitive closure is represented by a geometrically weighted
edgewise shared-partner statistic with fixed decay \(0.5\). The fitted model is
\[
Y\sim
\texttt{edges}
+\texttt{nodematch("Grade")}
+\texttt{gwesp}(0.5,\texttt{fixed=TRUE}).
\]

The six-chain population sampler uses base seed \(20260513\), \(20{,}000\)
burn-in iterations, \(5{,}000\) main iterations, and \(5{,}000\) auxiliary
ERGM iterations per exchange update. We set the differential-evolution weight
to \(\gamma=0.05\), initialize the penalty parameters at
\(\lambda_1=\lambda_2=5\), and use a scalar proposal variance of \(0.0025\).
The maximum pseudolikelihood estimate is used only to initialize the chains;
the sampler targets the full ERGM likelihood through the approximate exchange
updates. The Bayesian goodness-of-fit analysis uses \(500\)
posterior-predictive simulations and \(1{,}000\) auxiliary iterations.

Table~\ref{tab:fmh-elastic-net-posterior} gives the result.
The edge coefficient is strongly negative, reflecting the low baseline
propensity for friendship after conditioning on grade matching and shared
partners. The grade-matching and \texttt{gwesp} coefficients are both positive,
and their \(95\%\) confidence intervals lie well above zero.

\begin{table}[H]
\centering
\caption{Summaries for the BERGM Elastic Net fit to the
\texttt{faux.magnolia.high} network. The fitted model includes
\(\texttt{edges}\), \(\texttt{nodematch("Grade")}\), and
\(\texttt{gwesp}(0.5,\texttt{fixed=TRUE})\). The final two columns give the
empirical \(2.5\%\) and \(97.5\%\) quantiles.}
\label{tab:fmh-elastic-net-posterior}
\begin{tabular}{lrrrrr}
\toprule
Term & Mean & Median & SD & 2.5\% & 97.5\% \\
\midrule
\texttt{edges}
& -8.951 & -8.948 & 0.102 & -9.163 & -8.755 \\
\texttt{nodematch("Grade")}
& 3.068 & 3.067 & 0.119 & 2.846 & 3.309 \\
\texttt{gwesp.fixed.0.5}
& 2.359 & 2.334 & 0.302 & 1.862 & 2.998 \\
\bottomrule
\end{tabular}
\end{table}

The estimated grade-matching coefficient has mean \(3.068\).
Holding the other change statistics fixed, a one-unit increase in the
grade-matching change statistic multiplies the conditional odds of a tie by
$
\exp(3.068)=21.50.
$
Thus, the fitted model indicates a strong tendency toward friendships between
students in the same grade. The mean of the
\texttt{gwesp}(0.5) coefficient is \(2.359\), corresponding to a conditional
odds multiplier of
$
\exp(2.359)=10.58
$
for a one-unit increase in the associated change statistic. This positive
effect indicates that a potential friendship is more likely when the two
students share friends, after accounting for network density and grade
matching.

Figures~\ref{fig:fmh-plot-bergm} and \ref{fig:fmh-bgof} show the trace plots
and posterior-predictive checks. The main chains show the numerical
behavior of the adaptive fit, while the goodness-of-fit diagnostics compare
the observed network with networks generated from the fitted model. Taken
together, the coefficient summaries and predictive checks support a simple
description of this network: friendships are sparse overall, but ties are more
likely between students in the same grade and between students embedded in
shared-friend structures.

\begin{figure}[H]
    \centering
    \includegraphics[width=\textwidth]{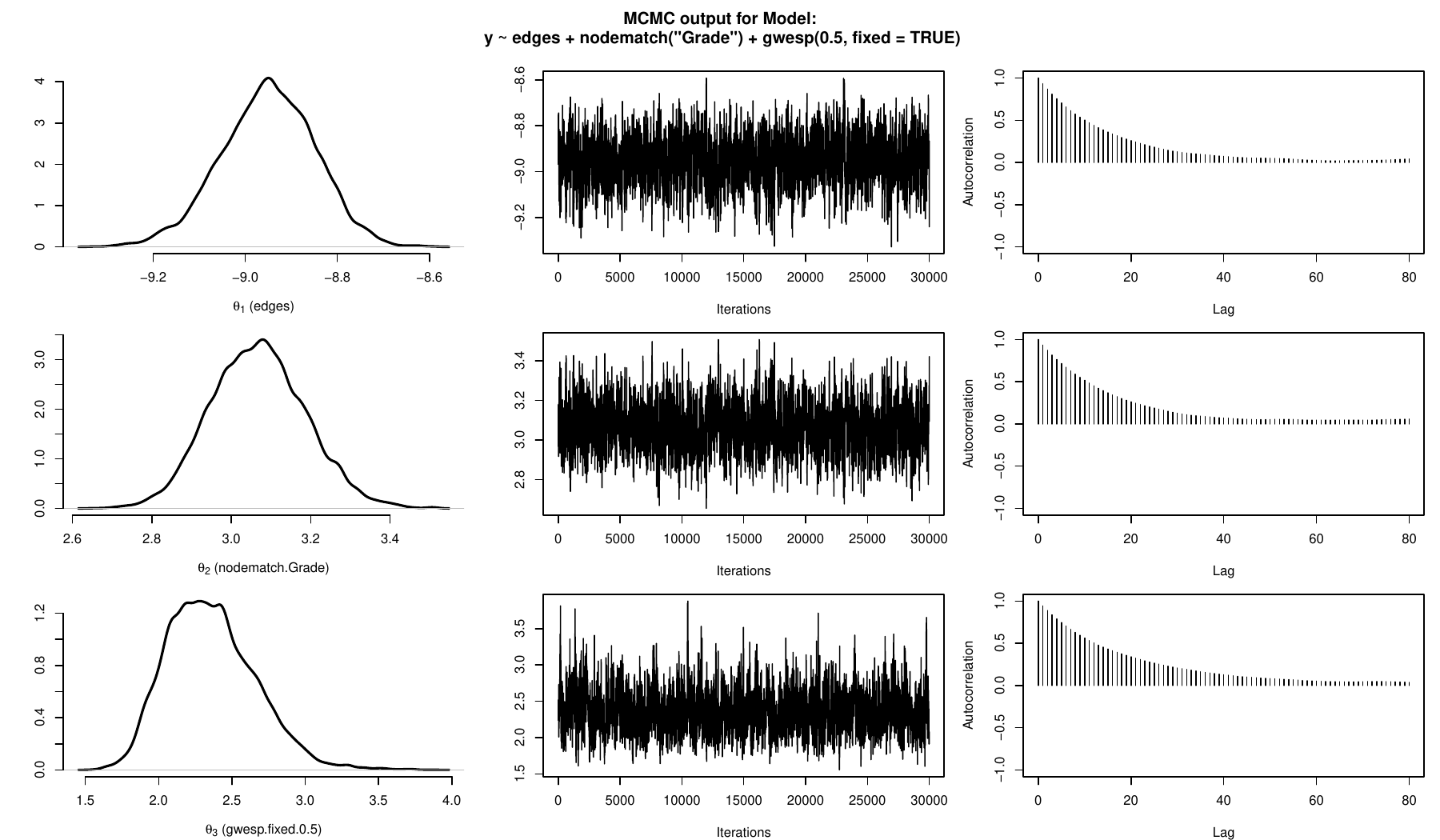}
    \caption{MCMC diagnostics for the Bayesian elastic-net BERGM fit to \texttt{faux.magnolia.high}.}
    \label{fig:fmh-plot-bergm}
\end{figure}

\begin{figure}[H]
    \centering
    \includegraphics[width=0.86\textwidth]{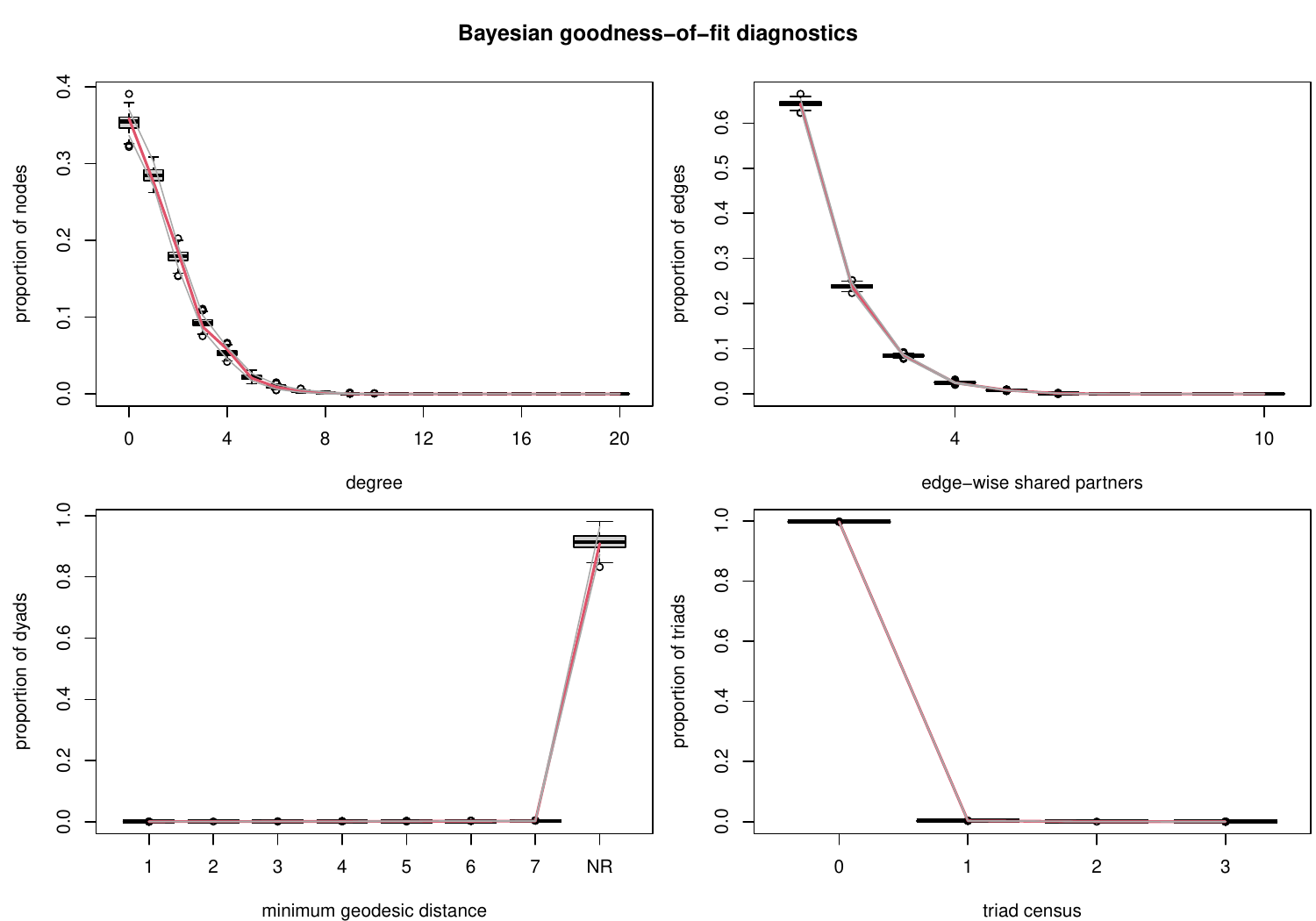}
    \caption{Bayesian goodness-of-fit diagnostics for the Bayesian elastic-net BERGM fit to \texttt{faux.magnolia.high}.}
    \label{fig:fmh-bgof}
\end{figure}

\subsection{OpenAlex AI Article Citation Network}
\label{subsec:openalex-ai}

Our second application considers a directed citation network of recent
artificial-intelligence articles from OpenAlex
\citep{PriemPiwowarOrr2022OpenAlex,OpenAlexDevelopers2026}. Within a selected
high-activity subgraph, we examine whether citation ties are associated with
topic and country similarity, local closure, target prominence, and source
reference breadth after accounting for directed degree heterogeneity. The
analysis illustrates the proposed BERGM Elastic Net on a directed network
substantially larger than those considered above.

\paragraph{Data and network construction}

\begin{figure}[h]
    \centering
    \includegraphics[width=\textwidth]
    {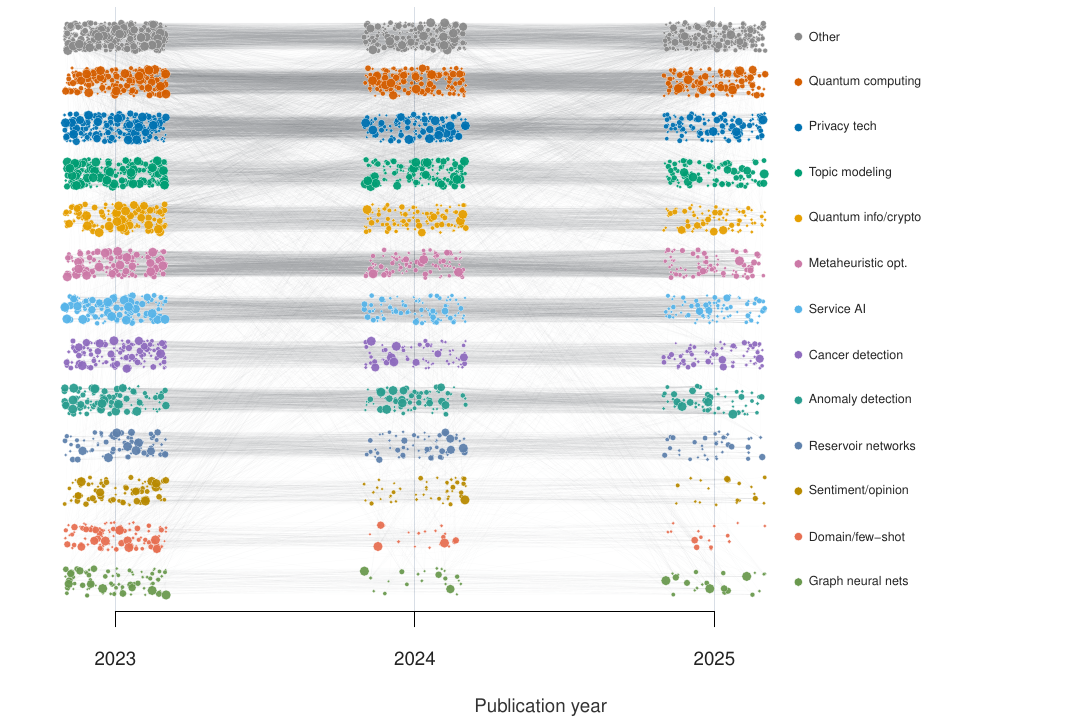}
    \caption{The OpenAlex AI citation network used in the analysis. The graph
    contains \(4{,}705\) articles and \(18{,}186\) directed citation ties.
    Nodes are colored by topic and scaled by within-source citation activity
    (\(\texttt{internal\_degree}\)), computed before subgraph selection;
    arrowheads are omitted for readability.}
    \label{fig:openalex-network}
\end{figure}

The data were downloaded from the OpenAlex Works API on June 19, 2026. The
source pool contains \(224{,}500\) articles published from 2023 through 2025
whose primary topic belongs to the OpenAlex Artificial Intelligence subfield
1702 and that had been cited by at least one OpenAlex Work
(\(\texttt{cited\_by\_count}>0\)). A directed edge \(i\to j\) indicates that
article \(i\) cites article \(j\).

Three article-level counts are used in constructing or modeling the network.
First, for each article we count the citations it sends to and receives from
other articles in the source pool (\(\texttt{internal\_degree}\)). Thus,
\(\texttt{internal\_degree}\) is the sum of the article's in-degree and
out-degree within the \(224{,}500\)-article source graph. Second, we record the
total number of OpenAlex Works that cite the article, including Works outside
the source pool (\(\texttt{cited\_by\_count}\)). This is a global citation
count at the date of extraction and is not the article's in-degree in the
fitted network. Third, we count the Work identifiers in the article's OpenAlex
reference list (\(\texttt{referenced\_work\_count}\)). This analysis variable
measures reference-list breadth as recorded by OpenAlex; it is not the
article's out-degree in the fitted network and need not equal the length of its
complete bibliography.

Among the source articles, \(141{,}176\) have at least one within-sample
citation, yielding \(344{,}816\) directed ties. Because this graph is too large
for repeated auxiliary simulation, we rank articles by
\(\texttt{internal\_degree}+\log\{1+\texttt{cited\_by\_count}\}\), induce a
subgraph on the top \(5{,}000\). The resulting network contains \(4{,}705\) articles and \(18{,}186\)
ties, with density \(0.00082\).

The remaining article attributes are publication year, lead country, and
primary topic. We have the 12 most frequent primary topics and combine the
remaining topics into the category \texttt{Other}. Missing lead-country values
are coded as \texttt{Unknown}. Figure~\ref{fig:openalex-network} displays the
analytic network.

\paragraph{Model specification}
The fitted ERGM includes an edge term, same-topic and same-country mixing,
topic-specific activity effects, absolute publication-year difference,
fixed-decay shared-partner closure, directed in- and out-degree controls,
target citation prominence, and source reference-list breadth. The edge and
matching terms describe baseline sparsity and citation concentration within
topic and country groups. Topic-activity effects are measured relative to
\texttt{Other}, and absolute publication-year difference is included as a
cross-sectional control. A fixed-decay geometrically weighted shared-partner
term represents local closure, while in- and out-degree count terms for degrees
\(1,\ldots,7\) control the marginal degree distributions. For details, see Table \ref{tab:openalex-key-effects}.

For a citation \(i\to j\), target citation prominence is the log-transformed
global number of OpenAlex Works citing article \(j\),
\(\log\{1+\texttt{cited\_by\_count}_j\}\). Source reference breadth is the
log-transformed number of OpenAlex Work identifiers in the reference list of
article \(i\),
\(\log\{1+\texttt{referenced\_work\_count}_i\}\). The first covariate describes
the global prominence of the cited article, whereas the second describes the
reference-list breadth of the citing article.

The adaptive fit used 20 chains, \(20{,}000\) burn-in iterations,
\(5{,}000\) main iterations, and \(20{,}000\) auxiliary ERGM iterations.
Table~\ref{tab:openalex-key-effects} reports
empirical-Bayes summaries.

\paragraph{Main findings}
The estimates portray a sparse but highly organized citation network. Against
the strongly negative edge baseline, topic similarity is the most pronounced
dyadic association. Holding the remaining change statistics fixed, changing a
dyad from cross-topic to same-topic multiplies its conditional citation odds by
\(\exp(3.013)=20.35\). This association remains large after allowing topics to
differ in their overall citation activity, indicating that it reflects
within-topic concentration rather than merely the greater activity of
particular fields. Citations are also geographically concentrated, although
less strongly: sharing the same lead country multiplies the conditional odds
by \(\exp(0.712)=2.04\).

\begin{table}[H]
\centering
\caption{Adaptive empirical-Bayes summaries for the OpenAlex
citation network. The 2.5\% and 97.5\% columns are empirical
quantiles. Topic-activity effects are measured relative to \texttt{Other};
degree-count terms control the marginal directed degree distributions.}
\label{tab:openalex-key-effects}
\footnotesize
\resizebox{\textwidth}{!}{%
\begin{tabular}{lrrrrr}
\toprule
Term & Mean & Median & SD & 2.5\% & 97.5\% \\
\midrule
\texttt{edges} & -14.665 & -14.655 & 0.242 & -15.121 & -14.189 \\
Same topic group & 3.013 & 3.013 & 0.060 & 2.899 & 3.132 \\
\addlinespace
\multicolumn{6}{l}{\textit{Topic-activity effects
(reference category: \texttt{Other})}} \\
Advanced Graph Neural Networks & 0.545 & 0.543 & 0.122 & 0.314 & 0.788 \\
AI in cancer detection & 0.187 & 0.186 & 0.071 & 0.050 & 0.320 \\
AI in Service Interactions & 0.024 & 0.024 & 0.067 & -0.102 & 0.155 \\
Anomaly Detection Techniques and Applications & 0.229 & 0.228 & 0.070 & 0.086 & 0.365 \\
Domain Adaptation and Few-Shot Learning & 0.306 & 0.308 & 0.115 & 0.074 & 0.539 \\
Metaheuristic Optimization Algorithms Research & 0.046 & 0.047 & 0.063 & -0.072 & 0.171 \\
Neural Networks and Reservoir Computing & 0.329 & 0.330 & 0.093 & 0.152 & 0.507 \\
Privacy-Preserving Technologies in Data & 0.137 & 0.138 & 0.055 & 0.026 & 0.240 \\
Quantum Computing Algorithms and Architecture & 0.039 & 0.040 & 0.053 & -0.067 & 0.141 \\
Quantum Information and Cryptography & 0.329 & 0.329 & 0.069 & 0.198 & 0.464 \\
Sentiment Analysis and Opinion Mining & 0.367 & 0.362 & 0.118 & 0.144 & 0.597 \\
Topic Modeling & 0.131 & 0.132 & 0.065 & 0.008 & 0.267 \\
\addlinespace
Same lead country & 0.712 & 0.709 & 0.060 & 0.598 & 0.830 \\
Absolute publication-year difference & 0.293 & 0.292 & 0.039 & 0.218 & 0.370 \\
Fixed-decay GWESP & 2.614 & 2.615 & 0.051 & 2.514 & 2.711 \\
\addlinespace
\multicolumn{6}{l}{\textit{Out-degree controls}} \\
\texttt{odegree(1)} & -1.407 & -1.404 & 0.094 & -1.589 & -1.227 \\
\texttt{odegree(2)} & -2.070 & -2.067 & 0.124 & -2.319 & -1.830 \\
\texttt{odegree(3)} & -2.037 & -2.032 & 0.152 & -2.340 & -1.751 \\
\texttt{odegree(4)} & -1.933 & -1.929 & 0.164 & -2.275 & -1.633 \\
\texttt{odegree(5)} & -1.513 & -1.508 & 0.157 & -1.840 & -1.213 \\
\texttt{odegree(6)} & -1.248 & -1.247 & 0.156 & -1.551 & -0.932 \\
\texttt{odegree(7)} & -0.789 & -0.789 & 0.122 & -1.027 & -0.552 \\
\addlinespace
\multicolumn{6}{l}{\textit{In-degree controls}} \\
\texttt{idegree(1)} & -0.334 & -0.333 & 0.100 & -0.533 & -0.144 \\
\texttt{idegree(2)} & -0.579 & -0.575 & 0.130 & -0.833 & -0.330 \\
\texttt{idegree(3)} & -0.698 & -0.697 & 0.133 & -0.942 & -0.429 \\
\texttt{idegree(4)} & -0.697 & -0.698 & 0.148 & -0.979 & -0.415 \\
\texttt{idegree(5)} & -0.509 & -0.510 & 0.151 & -0.800 & -0.218 \\
\texttt{idegree(6)} & -0.295 & -0.303 & 0.157 & -0.593 & 0.004 \\
\texttt{idegree(7)} & -0.150 & -0.140 & 0.139 & -0.424 & 0.104 \\
\addlinespace
Target citation prominence & 0.652 & 0.651 & 0.027 & 0.600 & 0.705 \\
Source reference breadth & 0.543 & 0.544 & 0.043 & 0.458 & 0.624 \\
\bottomrule
\end{tabular}%
}
\end{table}

The positive fixed-decay GWESP coefficient, \(2.614\), shows that citation
structure is not explained by observed article attributes alone. Ties tend to
occur within overlapping citation neighborhoods even after controlling for
topic, country, article activity, and directed degree heterogeneity. The
selected network therefore exhibits both substantive proximity and local
relational closure. Because the GWESP change statistic is nonlinear, however,
its coefficient is not an odds ratio for one additional shared partner.

The two article-level covariates describe opposite sides of a potential
citation \(i\to j\), where \(i\) is the citing article and \(j\) is the cited
article. Doubling \(1+\texttt{cited\_by\_count}_j\) multiplies the conditional
odds that \(i\) cites \(j\) by \(2^{0.652}=1.57\), holding the remaining change
statistics fixed. Doubling
\(1+\texttt{referenced\_work\_count}_i\) gives a corresponding multiplier of
\(2^{0.543}=1.46\). Thus, citation ties tend to point toward articles with
greater global citation prominence and to originate from articles with broader
OpenAlex-recorded reference lists. These are conditional associations within
the selected subgraph, not causal effects.

The topic-activity estimates further illustrate the role of elastic-net
regularization. Relative to \texttt{Other}, nine of the twelve topic groups
have confidence intervals entirely above zero. The clearest positive
contrast is Advanced Graph Neural Networks, with mean \(0.545\) and interval
\([0.314,0.788]\). In contrast, the intervals for AI in Service Interactions,
Metaheuristic Optimization Algorithms Research, and Quantum Computing
Algorithms and Architecture include zero, and their estimates remain close to the baseline. The fit therefore preserves pronounced topic-activity contrasts
while attenuating weaker ones. Because the prior is continuous, this pattern represents shrinkage rather than exact Bayesian selection, and it should not be interpreted as a ranking of AI research areas. Finally, the degree-count terms jointly account for residual heterogeneity in
the in- and out-degree distributions. Since a single edge toggle can change several of these statistics simultaneously, their coefficients are not interpreted separately.

\begin{figure}[h]
    \centering
    \includegraphics[width=0.88\textwidth]
    {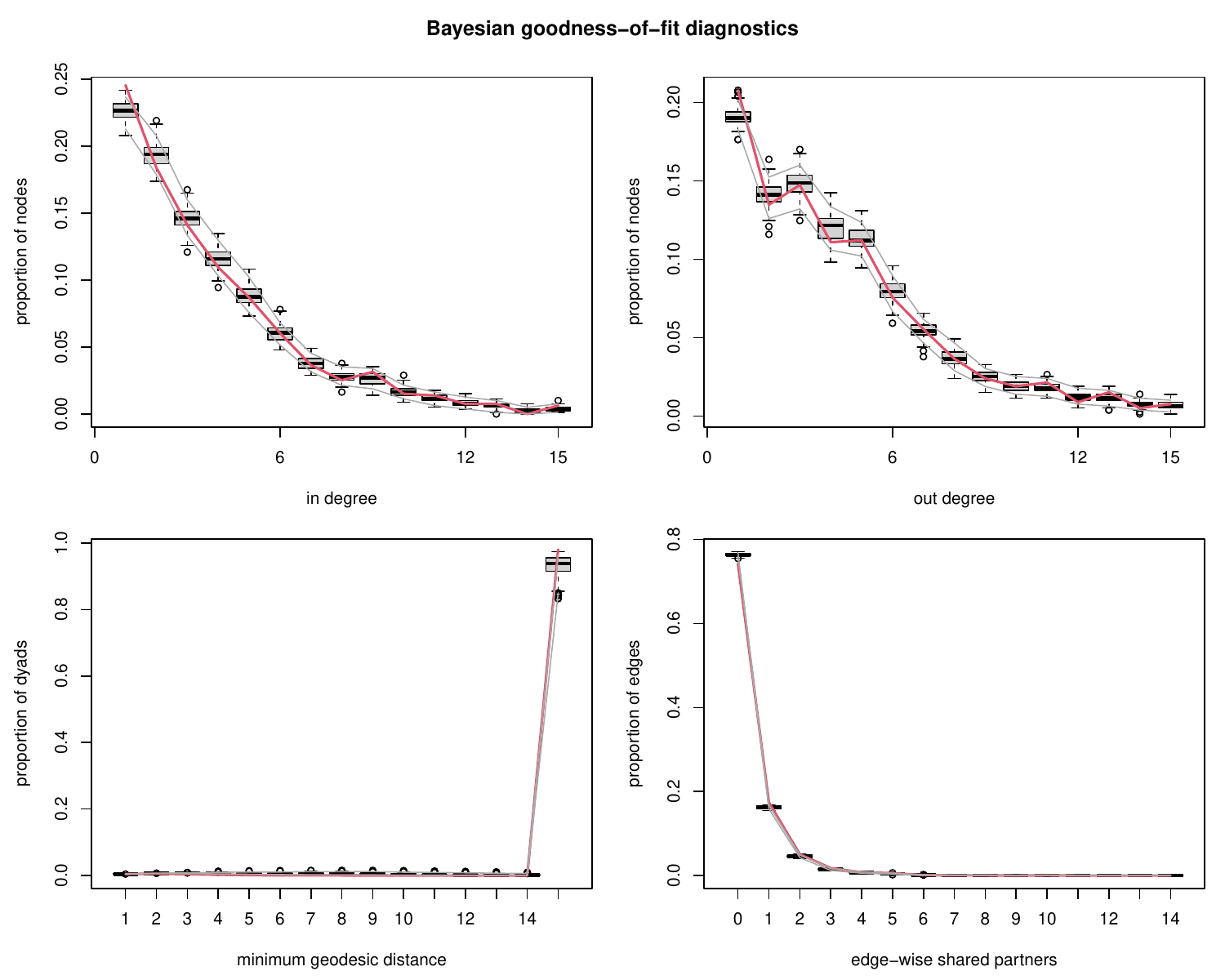}
    \caption{Approximate posterior-predictive diagnostics for the OpenAlex AI
    citation network based on \(1{,}000\) simulated networks. The panels
    compare observed and simulated in-degree, out-degree, minimum geodesic
    distance, and edgewise shared-partner summaries.}
    \label{fig:openalex-bgof}
\end{figure}

\paragraph{Predictive assessment.}
Figure~\ref{fig:openalex-bgof} compares the observed in-degree, out-degree,
geodesic distance, and edgewise shared-partner summaries with those from
\(1{,}000\) networks simulated using the model's estimation. The
displayed panels show no gross discrepancy for these coarse structural
features.

\section{Discussion}
\label{sec:discussion}

The results suggest that elastic-net regularization is particularly useful when an ERGM contains correlated signals together with many weak candidate terms. In the simulation, the proposed workflow substantially reduces estimation error and false-positive reporting while producing more similar estimates for the correlated active pair. This improvement comes with lower sensitivity under the conservative reporting rule, illustrating the central tradeoff
between suppressing noise and detecting weak signals. The two applications show that the same framework can be used in both small undirected networks and large directed networks with complex structural and nodal effects.

Future work may extend the framework to dynamic, multilayer, and valued networks and consider more flexible regularization structures for related network statistics. Further developments in scalable computation, uncertainty quantification, and model-selection theory would broaden the applicability of regularized Bayesian ERGMs to increasingly large and complex network data.

\section*{Data availability}

The \texttt{faux.magnolia.high} network is publicly available
through the \texttt{ergm} package, and the OpenAlex source records are publicly accessible through the OpenAlex Works API.  Analysis code and supporting materials are available in the Github xxxxxx.

\section*{Declaration of competing interest}

The authors declare that they have no known competing financial interests or personal relationships that could have appeared to influence the work reported in this paper.

\bibliographystyle{elsarticle-harv}
\bibliography{Bibliography-MM-MC.bib}

\section{Appendix}
\subsection{Proof of Lemma~\ref{lemma:normalizing_constant}}

The normalizing constant is determined by
\[
1
=
C(\lambda_1,\lambda_2)
\int_{-\infty}^{\infty}
\exp\left\{-\lambda_1|\theta_j|-\lambda_2\theta_j^2\right\}\,d\theta_j .
\]
By symmetry,
\[
\int_{-\infty}^{\infty}
\exp\left\{-\lambda_1|\theta_j|-\lambda_2\theta_j^2\right\}\,d\theta_j
=
2\int_0^\infty
\exp\left\{-\lambda_1\theta_j-\lambda_2\theta_j^2\right\}\,d\theta_j .
\]
Completing the square gives
\[
\lambda_2\theta_j^2
+
\lambda_1\theta_j
=
\lambda_2\left(\theta_j+\frac{\lambda_1}{2\lambda_2}\right)^2
-
\frac{\lambda_1^2}{4\lambda_2}.
\]
Hence
\[
\begin{aligned}
2\int_0^\infty
\exp\left\{-\lambda_1\theta_j-\lambda_2\theta_j^2\right\}\,d\theta_j
&=
2\exp\left\{\frac{\lambda_1^2}{4\lambda_2}\right\}
\int_0^\infty
\exp\left\{-\lambda_2\left(\theta_j+\frac{\lambda_1}{2\lambda_2}\right)^2\right\}d\theta_j .
\end{aligned}
\]
With
\[
u=\lambda_2\left(\theta_j+\frac{\lambda_1}{2\lambda_2}\right)^2,
\qquad
d\theta_j=\frac{1}{2\sqrt{\lambda_2}}u^{-1/2}\,du,
\]
the lower limit is \(u=\lambda_1^2/(4\lambda_2)\). Therefore
\[
\begin{aligned}
2\int_0^\infty
\exp\left\{-\lambda_1\theta_j-\lambda_2\theta_j^2\right\}\,d\theta_j
&=
\frac{1}{\sqrt{\lambda_2}}
\exp\left\{\frac{\lambda_1^2}{4\lambda_2}\right\}
\Gamma_U\left(\frac12,\frac{\lambda_1^2}{4\lambda_2}\right).
\end{aligned}
\]
Solving for \(C(\lambda_1,\lambda_2)\) gives
\[
C(\lambda_1,\lambda_2)
=
\sqrt{\lambda_2}
\exp\left\{-\frac{\lambda_1^2}{4\lambda_2}\right\}
\Gamma_U^{-1}\left(\frac12,\frac{\lambda_1^2}{4\lambda_2}\right),
\]
as claimed.

\subsection{Estimation of \texorpdfstring{$\lambda_1$ and $\lambda_2$}{lambda1 and lambda2}}

Up to an additive constant independent of \((\lambda_1,\lambda_2)\), the
complete-data log-likelihood is
\begin{align}
\ell_c(\lambda_1,\lambda_2 \mid \bs y,\th,\bs t)
&=
p\log \lambda_1
-p\log \Gamma_U\!\left(\frac12,\frac{\lambda_1^2}{4\lambda_2}\right)
-\lambda_2\sum_{j=1}^p \frac{t_j}{t_j-1}\theta_j^2
-\frac{\lambda_1^2}{4\lambda_2}\sum_{j=1}^p t_j .
\label{eq:complete_loglik_lambda_clean}
\end{align}

At the \(k\)th EM iteration, let
\[
\lambda^{(k-1)}
=
\bigl(\lambda_1^{(k-1)},\lambda_2^{(k-1)}\bigr).
\]
The corresponding \(Q\)-function is
\begin{align}
&Q(\lambda_1,\lambda_2\mid \lambda^{(k-1)})
=
\mathbb E\!\left[
\ell_c(\lambda_1,\lambda_2 \mid \bs y,\th,\bs t)
\;\middle|\;
\bs y,\lambda^{(k-1)}
\right]
\notag\\
&=
p\log \lambda_1
-
p\log \Gamma_U\!\left(
\frac12,\frac{\lambda_1^2}{4\lambda_2}
\right)
-\lambda_2\,
\mathbb E\!\left[
\sum_{j=1}^p \frac{t_j}{t_j-1}\theta_j^2
\;\middle|\;
\bs y,\lambda^{(k-1)}
\right]
\notag\\
&\quad
-\frac{\lambda_1^2}{4\lambda_2}\,
\mathbb E\!\left[
\sum_{j=1}^p t_j
\;\middle|\;
\bs y,\lambda^{(k-1)}
\right].
\label{eq:Q_function_lambda_clean}
\end{align}

For notational convenience, define
\begin{align}
A^{(k-1)}
&=
\mathbb E\!\left[
\sum_{j=1}^p \frac{t_j}{t_j-1}\theta_j^2
\;\middle|\;
\bs y,\lambda^{(k-1)}
\right],
\label{eq:A_kminus1}
\\
B^{(k-1)}
&=
\mathbb E\!\left[
\sum_{j=1}^p t_j
\;\middle|\;
\bs y,\lambda^{(k-1)}
\right].
\label{eq:B_kminus1}
\end{align}
Further, define
\begin{align}
h(\lambda_1,\lambda_2)
&=
\frac{\exp\{-\delta(\lambda_1,\lambda_2)\}}
{\Gamma_U\!\left(\frac12,\delta(\lambda_1,\lambda_2)\right)
\sqrt{\delta(\lambda_1,\lambda_2)}},
\label{eq:h_definition}
\\
\delta(\lambda_1,\lambda_2)
&=
\frac{\lambda_1^2}{4\lambda_2}.
\end{align}
Since
\[
\frac{\partial}{\partial x}\Gamma_U\!\left(\frac12,x\right)
=
-x^{-1/2}e^{-x},
\]
the score equations are
\begin{align}
\frac{\partial Q}{\partial \lambda_1}
&=
\frac{p}{\lambda_1}
+
p\,\frac{\lambda_1}{2\lambda_2}\,
h(\lambda_1,\lambda_2)
-
\frac{\lambda_1}{2\lambda_2}\,B^{(k-1)}=0,
\label{eq:score_lambda1_clean}
\\[1mm]
\frac{\partial Q}{\partial \lambda_2}
&=
-
p\,\frac{\lambda_1^2}{4\lambda_2^2}\,
h(\lambda_1,\lambda_2)
-
A^{(k-1)}
+
\frac{\lambda_1^2}{4\lambda_2^2}\,B^{(k-1)}=0.
\label{eq:score_lambda2_clean}
\end{align}

In practice, the expectations in \eqref{eq:A_kminus1}--\eqref{eq:B_kminus1}
are not available in closed form and are replaced by Monte Carlo averages.
Suppose
\[
\bigl\{(\th^{(m)},\bs t^{(m)})\bigr\}_{m=1}^M
\sim
\pi(\th,\bs t\mid \bs y,\lambda^{(k-1)}).
\]
Then define
\begin{align}
A_M
&=
\frac{1}{M}\sum_{m=1}^M
\sum_{j=1}^p
\frac{t_j^{(m)}}{t_j^{(m)}-1}\bigl(\theta_j^{(m)}\bigr)^2,
\\
B_M
&=
\frac{1}{M}\sum_{m=1}^M
\sum_{j=1}^p t_j^{(m)},
\\
\delta
&=
\frac{\lambda_1^2}{4\lambda_2},
\\
h
&=
\frac{e^{-\delta}}
{\Gamma_U\!\left(\frac12,\delta\right)\sqrt{\delta}}.
\end{align}
Accordingly, the Monte Carlo approximation of the \(Q\)-function is
\begin{align}
\widehat Q_M(\lambda_1,\lambda_2)
&=
p\log\lambda_1
-p\log \Gamma_U\!\left(\frac12,\delta\right)
-\lambda_2 A_M
-\frac{\lambda_1^2}{4\lambda_2}B_M.
\label{eq:Qhat_clean}
\end{align}
The corresponding Monte Carlo score equations are
\begin{align}
\frac{\partial \widehat Q_M}{\partial \lambda_1}
&=
\frac{p}{\lambda_1}
+p\,\frac{\lambda_1}{2\lambda_2}\,h
-\frac{\lambda_1}{2\lambda_2}B_M
=0,
\label{eq:Qhat_score_lambda1}
\\[1mm]
\frac{\partial \widehat Q_M}{\partial \lambda_2}
&=
-p\,\frac{\lambda_1^2}{4\lambda_2^2}\,h
-A_M
+\frac{\lambda_1^2}{4\lambda_2^2}B_M
=0.
\label{eq:Qhat_score_lambda2}
\end{align}
As above, combining \eqref{eq:Qhat_score_lambda1} and
\eqref{eq:Qhat_score_lambda2} yields the closed-form Monte Carlo update
\begin{equation}
\lambda_2^{(k)}
=
\frac{p}{2A_M},
\label{eq:lambda2_MC_update_clean}
\end{equation}
while \(\lambda_1^{(k)}\) is obtained by solving
\begin{equation}
\frac{p}{\lambda_1}
+p\,\frac{\lambda_1}{2\lambda_2^{(k)}}\,h
-\frac{\lambda_1}{2\lambda_2^{(k)}}B_M
=0
\label{eq:lambda1_MC_equation_clean}
\end{equation}
numerically.

Therefore, each MCEM iteration proceeds as follows.

\medskip
\noindent
\textbf{E-step.}
Sample
\[
(\th,\bs t)\mid \bs y,\lambda^{(k-1)}
\]
from the current posterior distribution and compute the Monte Carlo summaries
\(A_M\) and \(B_M\).

\medskip
\noindent
\textbf{M-step.}
Update \(\lambda_2^{(k)}\) by \eqref{eq:lambda2_MC_update_clean}, and then solve
\eqref{eq:lambda1_MC_equation_clean} numerically to obtain
\(\lambda_1^{(k)}\).

\medskip

This procedure yields empirical Bayes estimates of the two elastic-net
hyperparameters while preserving the hierarchical Bayesian structure of the
model.

\subsection{Posterior propriety}\label{app:posterior propriety}

\begin{theorem}[Posterior propriety]
\label{thm:posterior-propriety}
Fix \(\lambda_1>0\) and \(\lambda_2>0\), and suppose that the ERGM sample space
\(\mathcal Y\) is finite. Then for every observed network
\(\bs y\in\mathcal Y\), the posterior density
\[
\pi(\th\mid \bs y,\lambda_1,\lambda_2)
\propto
\frac{\exp\{\th^\top s(\bs y)\}}{z(\th)}
\exp\{-\lambda_1\|\th\|_1-\lambda_2\|\th\|_2^2\}
\]
is proper.
\end{theorem}

\begin{proof}
The elastic-net prior is proper by
Lemma~\ref{lemma:normalizing_constant}. For every \(\th\),
\(\pi(\bs y\mid\th)\) is a probability mass function evaluated at \(\bs y\),
so \(0<\pi(\bs y\mid\th)\le1\). Hence the marginal likelihood is finite and
strictly positive:
\[
0<
\int \pi(\bs y\mid\th)\pi(\th\mid\lambda_1,\lambda_2)\,d\th
\le
\int \pi(\th\mid\lambda_1,\lambda_2)\,d\th
=1.
\]
\end{proof}

\subsection{Proof of the Bayes reporting rule}
\label{app:proof-reporting}

\begin{proof}[Proof of Proposition~\ref{prop:reporting-decision}]
For a fixed-hyperparameter posterior, the posterior risk of reporting term
\(j\) as active is
\[
R(d_j=1\mid\bs y)
=
c_{\mathrm{FP}}\{1-p_j(\delta)\}.
\]
The posterior risk of not reporting it is
\[
R(d_j=0\mid\bs y)
=
c_{\mathrm{FN}}p_j(\delta).
\]
The Bayes action reports the term when
\[
c_{\mathrm{FP}}\{1-p_j(\delta)\}
<
c_{\mathrm{FN}}p_j(\delta).
\]
Rearranging this inequality gives
\[
p_j(\delta)
>
\frac{c_{\mathrm{FP}}}
     {c_{\mathrm{FP}}+c_{\mathrm{FN}}}.
\]
At equality, both actions have the same posterior risk.

For the cutoff used in the simulation,
\[
0.90
=
\frac{c_{\mathrm{FP}}}
     {c_{\mathrm{FP}}+c_{\mathrm{FN}}},
\]
which is equivalent to
\[
c_{\mathrm{FP}}=9c_{\mathrm{FN}}.
\]
\end{proof}

\subsection{Derivation of the pseudolikelihood grouping bound}
\label{app:pl-grouping-proof}

The elastic-net penalized pseudolikelihood estimator minimizes
\[
Q_{\mathrm{PL}}(\th)
=
\sum_{m=1}^{M}
\left[
\log\{1+\exp(x_m^\top\th)\}
-
y_mx_m^\top\th
\right]
+
\lambda_1\|\th\|_1
+
\lambda_2\|\th\|_2^2.
\]
Define
\[
\mu_m(\th)
=
\frac{\exp(x_m^\top\th)}
     {1+\exp(x_m^\top\th)}.
\]
For a nonzero coordinate
\(\widehat\theta_j^{\mathrm{PL}}\), the KKT equation is
\[
\sum_{m=1}^{M}
\{\mu_m(\widehat\th^{\mathrm{PL}})-y_m\}x_{mj}
+
\lambda_1
\operatorname{sign}(\widehat\theta_j^{\mathrm{PL}})
+
2\lambda_2\widehat\theta_j^{\mathrm{PL}}
=
0.
\]
If coordinates \(j\) and \(k\) are nonzero and have the same sign, subtracting
their KKT equations eliminates the \(L_1\) terms and gives
\[
2\lambda_2
\left(
\widehat\theta_j^{\mathrm{PL}}
-
\widehat\theta_k^{\mathrm{PL}}
\right)
=
\sum_{m=1}^{M}
\{y_m-\mu_m(\widehat\th^{\mathrm{PL}})\}
(x_{mj}-x_{mk}).
\]
Since
\[
|y_m-\mu_m(\widehat\th^{\mathrm{PL}})|\leq1,
\]
we obtain
\[
2\lambda_2
\left|
\widehat\theta_j^{\mathrm{PL}}
-
\widehat\theta_k^{\mathrm{PL}}
\right|
\leq
\sum_{m=1}^{M}|x_{mj}-x_{mk}|
=
\|X_j-X_k\|_1.
\]
Therefore,
\[
\left|
\widehat\theta_j^{\mathrm{PL}}
-
\widehat\theta_k^{\mathrm{PL}}
\right|
\leq
\frac{\|X_j-X_k\|_1}{2\lambda_2}.
\]
The Cauchy--Schwarz inequality further gives
\[
\|X_j-X_k\|_1
\leq
\sqrt M\,\|X_j-X_k\|_2,
\]
and hence
\[
\left|
\widehat\theta_j^{\mathrm{PL}}
-
\widehat\theta_k^{\mathrm{PL}}
\right|
\leq
\frac{\sqrt M\,\|X_j-X_k\|_2}{2\lambda_2}.
\]
This proves Proposition~\ref{prop:pl-grouping}.

If the two design columns are centered and standardized so that
\[
\|X_j\|_2^2=\|X_k\|_2^2=M
\]
and their sample correlation is \(r_{jk}\), then
\[
X_j^\top X_k=Mr_{jk}.
\]
Consequently,
\[
\begin{aligned}
\|X_j-X_k\|_2^2
&=
\|X_j\|_2^2+\|X_k\|_2^2-2X_j^\top X_k \\
&=
2M(1-r_{jk}).
\end{aligned}
\]
Thus, high positive correlation between the pseudolikelihood design columns
tightens the coefficient bound. This calculation applies to the
pseudolikelihood approximation and not directly to the full ERGM likelihood.

\subsection{Proof of the full-likelihood grouping theorem}
\label{app:full-likelihood-grouping-proof}

\begin{proof}[Proof of
Theorem~\ref{thm:full-likelihood-grouping-effect}]
Consider the full-likelihood elastic-net objective
\[
\Phi(\th)
=
\log z(\th)-\th^\top s(\bs y)
+
\lambda_1\|\th\|_1
+
\lambda_2\|\th\|_2^2.
\]
Because \(\mathcal Y\) is finite, \(\log z(\th)\) is finite, smooth, and
convex. Since \(\lambda_2>0\), the objective \(\Phi\) is
\(2\lambda_2\)-strongly convex and coercive. It therefore has a unique
minimizer \(\widehat\th\).

Differentiating the log-partition function gives
\[
\frac{\partial}{\partial\theta_j}\log z(\th)
=
\mathbb E_{\th}\{s_j(\bs Y)\}.
\]
For a nonzero MAP coordinate \(\widehat\theta_j\), the KKT equation is
\[
\mathbb E_{\widehat\th}\{s_j(\bs Y)\}
-
s_j(\bs y)
+
\lambda_1\operatorname{sign}(\widehat\theta_j)
+
2\lambda_2\widehat\theta_j
=
0.
\]
Similarly, for coordinate \(k\),
\[
\mathbb E_{\widehat\th}\{s_k(\bs Y)\}
-
s_k(\bs y)
+
\lambda_1\operatorname{sign}(\widehat\theta_k)
+
2\lambda_2\widehat\theta_k
=
0.
\]
Because the two coordinates have the same sign, subtracting these equations
eliminates the \(L_1\) terms. It follows that
\[
2\lambda_2(\widehat\theta_j-\widehat\theta_k)
=
\{s_j(\bs y)-s_k(\bs y)\}
-
\mathbb E_{\widehat\th}
\{s_j(\bs Y)-s_k(\bs Y)\},
\]
or, equivalently,
\[
2\lambda_2(\widehat\theta_j-\widehat\theta_k)
=
g_{jk}(\bs y)
-
\mathbb E_{\widehat\th}\{g_{jk}(\bs Y)\}.
\]

Since \(\mathcal Y\) is finite, both
\(g_{jk}(\bs y)\) and
\(\mathbb E_{\widehat\th}\{g_{jk}(\bs Y)\}\) lie in the interval
\[
\left[
\min_{\bs u\in\mathcal Y}g_{jk}(\bs u),
\,
\max_{\bs u\in\mathcal Y}g_{jk}(\bs u)
\right].
\]
Therefore,
\[
\left|
g_{jk}(\bs y)
-
\mathbb E_{\widehat\th}\{g_{jk}(\bs Y)\}
\right|
\leq
D_{jk}.
\]
Dividing by \(2\lambda_2\) yields
\[
|\widehat\theta_j-\widehat\theta_k|
\leq
\frac{D_{jk}}{2\lambda_2}.
\]
If \(s_j-s_k\) is constant on \(\mathcal Y\), then \(D_{jk}=0\), and hence
\[
\widehat\theta_j=\widehat\theta_k.
\]
\end{proof}

\end{document}